\documentclass[draftcls,onecolumn,12pt, draftclsnofoot]{IEEEtran}

\normalsize
\usepackage{amsmath, amsfonts, amssymb, amsthm}

\usepackage{optidef}
\usepackage{color}
\usepackage{bm}
\usepackage{tikz,pgfplots}
\usetikzlibrary{patterns}
\usepackage{algorithm}
\usepackage{cite}
\usepackage{algpseudocode}
\usepackage{siunitx}
\usepackage{subcaption}
\usepackage{standalone}
\usepackage{enumitem}
\usepackage{microtype}
\usepackage[implicit=false]{hyperref}

\newcommand{\vardbtilde}[1]{\tilde{\raisebox{0pt}[0.85\height]{$\tilde{#1}$}}}

\newtheoremstyle{note}
{3pt}
{3pt}
{}
{}
{\bfseries\upshape}
{.}
{.5em}
{}

\theoremstyle{note}
\newtheorem{lemma}{Lemma}
\newtheorem{theorem}{Theorem}

\usetikzlibrary{arrows}
\usetikzlibrary{arrows,decorations.pathmorphing,backgrounds,positioning,fit,petri,shapes.geometric}

\begin{document}
\title{On Energy Allocation and Data Scheduling  in Backscatter Networks with Multi-antenna Readers }
	
\author{Mohammad~Movahednasab,  Mohammad~Reza~Pakravan,~\IEEEmembership{Member,~IEEE,} Behrooz~Makki,~\IEEEmembership{Senior~Member,~IEEE,} and Tommy~Svensson,~\IEEEmembership{Senior~Member,~IEEE}

\thanks{M. Movahednasab and M. R. Pakravan are with  Sharif University of Technology, Tehran, Iran (Emails: movahednasab@ee.sharif.edu, pakravan@sharif.edu).
Behrooz Makki is in Gothenburg, Sweden (Email: behrooz.makki@gmail.com).
Tommy Svensson is with Chalmers University of Technology, Gothenburg, Sweden (Email: 
tommy.svensson@chalmers.se).}}

\maketitle

\begin{abstract}
In this paper, we study the throughput utility functions in buffer-equipped monostatic backscatter communication networks with multi-antenna Readers. In the considered model, the backscatter nodes (BNs) store the data in their buffers before transmission to the Reader. We investigate three utility functions, namely, the sum, the proportional and the common throughput. We design online admission policies, corresponding to each utility function, to determine how much data can be admitted in the buffers. Moreover, we propose an online data link control policy for jointly controlling the transmit and receive beamforming vectors as well as the reflection coefficients of the BNs. The proposed policies for data admission and data link control jointly optimize the throughput utility, while stabilizing the buffers. We adopt the min-drift-plus-penalty (MDPP) method in designing the control policies. Following the MDPP method, we cast the optimal data link control and the data admission policies as solutions of two independent optimization problems which should be solved in each time slot. The optimization problem corresponding to the data link control is non-convex and does not have a trivial solution. Using Lagrangian dual  and quadratic transforms, we find a closed-form iterative solution. Finally, we use the results on the achievable rates of finite blocklength codes to study the system performance in the cases with short packets. As demonstrated, the proposed policies achieve optimal utility and stabilize the data buffers in the BNs. 
\end{abstract}
\begin{IEEEkeywords}
Backscatter communication, radio frequency identification, fairness, min-drift-plus-penalty, Lyapunov optimization, max-min throughput, proportional throughput, sum throughput, finite blocklength analysis,  wireless energy transfer, energy harvesting, green communications, Internet of things, IoT.
\end{IEEEkeywords}

\section{INTRODUCTION}

\IEEEPARstart{W}{ith} the emergence of the Internet of things (IoT) era, the number of connected devices is increasing rapidly. It is predicted  that there will be around 5 billion connected devices in 2025 \cite{Ericsson}, a great number of which will be portable and low-power. This explosion of the low-powered devices calls for replacing batteries with new energy  sources to ensure  continuous operation of devices. The major challenges with the battery-powered devices are the increase of the devices' form factor and the high cost for recharging and replacement of the batteries \cite{FRezaei}. Moreover, in some applications such as biomedical implants inside human bodies or distributed monitoring sensors, replacing the batteries may be infeasible \cite{Zeng, Han, Liu2019}. 

Backscatter communication networks (BCNs) are considered to be a prominent solution for low-power and low-cost communications.  A BCN compromises a Reader and, possibly, multiple backscattering nodes (BNs) with most  bulky and active communication modules moved to the Reader.  The BNs transmit data to the Reader via reflecting and modulating the incident radio frequency signal by adapting the level of antenna mismatch to vary the reflection coefficient \cite{Han}. Based on the source of the radio frequency signal, which supplies the required energy for communication, three  configurations  for the  BCN, namely, monostatic, bistatic and ambient,  are considered. In monostatic configuration, the Reader emits a carrier and receives the backscattered data, while in bistatic configuration one or several power beacons emit carries rather than the Reader itself \cite{Kimionis}. Moreover, in ambient BCNs there is no dedicated energy transmitter and the BNs backscatter the existing radio frequency signals, e.g., WiFi or digital television signals \cite{Huynh}. 

 The low-energy transmission efficiency is a major problem in BCNs.  However, the studies in, e.g.,  \cite{Yang, Chen, Mishra_Sum, Mishra_Min, Yang2018,Idrees2020,Li2018,Tao, ChenChen1, ChenChen2, Boyer, He2020}, show that exploiting multiple antennas  increases the energy efficiency remarkably. Considering the energy required for channel training,  \cite{Yang} optimizes the transmit  beamforming to maximize the harvested energy by the BNs. In \cite{Chen}, a blind adaptive beamforming scheme is introduced to increase the reading range of the radio frequency identification tags. Also, \cite{Mishra_Sum} and \cite{Mishra_Min}  propose low complexity algorithms for optimizing the transmit and receive beamforming to maximize the sum and  minimum throughput of all BNs, respectively. Considering ambient BCNs,  \cite{Yang2018, Idrees2020, Li2018,Tao} study the receive beamforming optimization of a multi-antenna Reader. Whereas,  \cite{ChenChen1} and \cite {ChenChen2}  study optimal detectors for ambient BCNs with multi-antenna BNs and a single antenna Reader. In \cite{Boyer} and \cite{He2020}, the  achievable diversity order of a multiple-input-multiple-output monostatic BCN is studied. Furthermore, optimizing the reflection coefficients of the BNs is studied in, e.g., \cite{Xiao2019,Ye2019, Yang2020, Krikidis2018,Yang2019,Gong2018}, where maximizing energy efficiency, throughput or  fairness in the BCNs is investigated.

Along with energy efficiency, one of the challenges of the BCNs is the unpredictability of the channel state and the available energy in ambient  configuration, which makes the optimal scheduling difficult. To tackle this problem,  stochastic approaches are adopted in \cite{Wen2019,   Huynh2019,   Anh2019,     Hoang2017, Liu2020} to design  online control algorithms. Specifically, in \cite{Wen2019, Huynh2019,   Anh2019}, long-term throughput optimization in ambient BCNs is studied through reinforcement learning methods. 
Whereas, in \cite{Hoang2017}  the authors use reinforcement learning to propose a data admission  and data scheduling policy for a monostatic BCN. Finally, \cite{Liu2020} uses min-drift-plus-penalty (MDPP) method to maximize throughput in an ambient BCN. 

In this work, we concentrate on optimizing different throughput utility functions  in monostatic BCNs with multi-antenna Readers. In our considered model, the BNs store the data in their buffers before transmission to the Reader. We investigate three different utility functions including the sum, the proportional and the common throughput. We design three online admission policies, corresponding to  each throughput utility function, to determine how much data can be admitted in the buffers in each time slot. Moreover, we propose an online  data link control policy for jointly controlling the transmit and receive beamforming vectors as well as the reflection coefficients of the BNs. {The proposed policies for data admission and data link control jointly optimize the throughput utility, while stabilizing the buffers.}
\color{black}

Considering the channel state randomness, we adopt the MDPP  method in designing the control policies. Following the MDPP method, we cast the optimal data link control  and the data admission policies as solutions of two independent optimization problems which should be solved in each time slot.  The optimization problem corresponding to the data link control is non-convex and does not have a trivial solution. We transform this problem into an equivalent form, which  makes it possible to find a closed-form iterative algorithm that improves the utility in each iteration. Furthermore, considering each utility function, we solve the corresponding data admission problem and find closed-form solutions. Finally, we use the results on the achievable rates of finite blocklength codes \cite{Polyanskiy, Makki2014, Makki2016, Haghifam} to study the system performance in the cases with short packets.

The differences in the considered system model and problem formulation   makes our paper different from those in the literature. Specifically, this paper is different from \cite{Yang2018, Idrees2020, Li2018,Tao, ChenChen1, ChenChen2,   Xiao2019 ,Ye2019,   Yang2020,    Yang2019,      Wen2019, Huynh2019,   Anh2019,Hoang2017,Liu2020}, because we study a monostatic BCN, where we jointly optimize the transmit and receive beamforming vectors. As opposed to \cite{Yang, Chen, Mishra_Sum, Mishra_Min, Yang2018,Idrees2020,Li2018,Tao, ChenChen1, ChenChen2, Boyer, He2020, Xiao2019,Ye2019, Yang2020, Krikidis2018,Yang2019,Gong2018}, we consider BNs equipped with buffers, and hence follow  an stochastic approach to propose an online policy to maximize the throughput utilities and stabilize the buffers.  Moreover, different from the state-of-the-art works, we study different throughput utilities in a unified framewok, perform finite block-length analysis, and compare the results in terms of fairness.

Our analytical and simulations results show the significance of our proposed control policy. Specifically, we show in Theorem \ref{th:optimal} that, under the optimal solutions of the link scheduling and data admission problems, the average throughput utility will be within $\mathcal{O}({1\over V})$ of the optimal utility, where $V$ is a control parameter. Moreover, the level of the stored data in the buffers will be kept under a certain level of $\mathcal{O}({V})$. Furthermore, we compare our proposed scheme with the benchmarks  \cite{Mishra_Sum} and \cite{Mishra_Min}, which verifies the optimality of our proposed policy. Particularly, \cite{Mishra_Sum} and \cite{Mishra_Min} propose optimal control policies to maximize the sum and the minimum channel rate among BNs in a time slot base, respectively. However, the BNs in these two papers do not consider buffers. We show that the sum throughput utility under our proposed policy achieves the  maximum sum channel rate in  \cite{Mishra_Sum}, while stabilizing the data buffers. Whereas, our proposed policy for optimizing the common throughput improves the result in \cite{Mishra_Min}, which shows the significant effect of adopting buffers in the BNs.

The rest of the paper is organized as follows. The considered system model and our problem formulation are illustrated in Section \ref{sec:sysModel}. The proposed control policy as well as its performance analysis are presented in Section \ref{sec:propPolicy}. Simulation results are presented in Section \ref{sec:Numerical}, and finally, Section \ref{sec:conclude} concludes the paper.

Notation: Matrices and vectors are denoted by small and
capital boldface letters, respectively. Moreover, unless otherwise
mentioned, vectors are single-column matrices. Also, $(.)^T$ , $(.)^H$
and $(.)^\ast$ denote transpose, conjugate transpose and elementwise conjugate of a matrix, respectively.  Then,  $|.|$ denotes
the absolute value (or the modulus for complex numbers) and $\lVert.\lVert$ denotes the norm of vectors. Finally, $\mathbb{R}$ and $\mathbb{C}$ represent the real and complex number sets, respectively, $\operatorname{Re}\{.\}$ is used to denote the real part of a complex number, $\operatorname{E}\{.\}$ is the expectation, { $I_{N}$ represents the $N\times N$ identity matrix}, and $[a]^+ = \max\{a, 0\}, \;\forall a \in \mathbb{R}$.

\section{System Model and Problem Formulation}\label{sec:sysModel}
An example of the considered network structure is depicted in Fig. \ref{fig:sampNet}. We consider a BCN consisting of a multi-antenna Reader and $N$ single-antenna BNs.  Let  $\text{BN}_n, \;\forall n \in\mathcal{N} \triangleq\{1,\ldots,N\}$, denote the $n$-th BN in the network. The BNs  transmit the  data stored in their buffers to the Reader through backscattering the energy emitted by the Reader. While the BNs rely on the Reader energy for  data transmission, they have an internal battery which powers up their low-power circuits. Accordingly, the BNs are semi-passive devices and are able to handle data sensing or other internal tasks without  the Reader energy \cite{Vannucci2008,Lu2018,Fasarakis2015}. 
The Reader is equipped with $M$ antennas to focus its emitted energy towards the intended BNs. Moreover, the presence of multiple antennas enables the Reader to  receive data from multiple BNs through receive beamforming technique. The Reader adopts a \textit{kind of} full duplex structure and uses a common set of antennas in the transmit and receive path. The transmitted carrier and the received backscatter signal are then separated  using circulators. It is important to note that, unlike the conventional full-duplex wireless links, which transmit a modulated signal, in BCNs the unmodulated carrier leakage to the receiver can be  efficiently surpassed with low complexity methods \cite{Villame2010, Hao2018}.   
\begin{figure}
\centering
  \includegraphics[width = 0.6\textwidth]{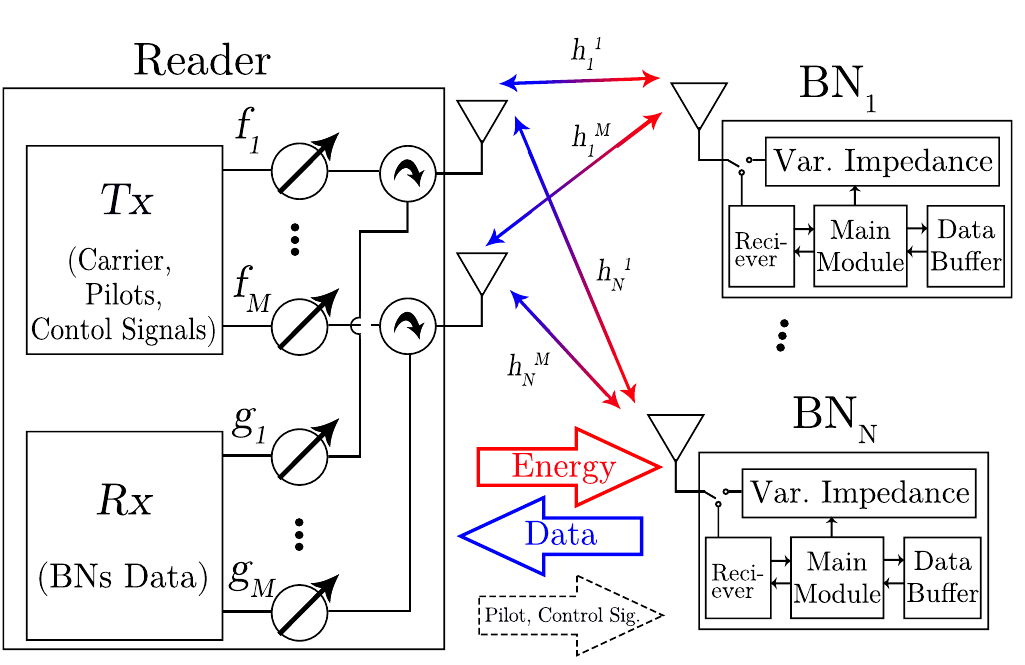}
\caption{An example of the Reader and BNs structure. The Reader transmits energy in the downlink, while the BNs transmit data through backscattering the received energy by changing their variable impedance. During the scheduling time, the BNs switch to receive mode and the Reader transmits control signals.}
\label{fig:sampNet}
\end{figure}

The time horizon is slotted to fixed duration time slots.  At the beginning of each time slot, a small portion of the time is devoted to channel estimation and node discovery. Note that our proposed policy only requires the end-to-end CSI, $\bm{h}_n(t)\bm{h}^T_n(t)$, which simplifies the CSI estimation process. The Reader then schedules its transmit and receive beamforming for the rest of the time slot. Furthermore, the Reader notifies the BNs how much new data they can admit in their buffers and also  sets the values of the reflection coefficients of the BNs. In the rest of the time slot, the Reader transmits energy, and concurrently  receives the backscattered signal from the BNs. 
The channel coefficients are assumed to be constant in a time slot but vary randomly and independently in consecutive time slots.  In  time slot $t$, $h_n^m(t) \in \mathbb{C}, m = 1,\ldots,M$, denotes the channel gain of the reciprocal link between the $m$-th antenna of the Reader and  $\text{BN}_n$, and $\bm{h}_n(t) = (h_n^1(t), \ldots, h_n^m(t))^T$ is the channel coefficients vector associated with $\text{BN}_n$.

\begin{figure}
\centering
\includegraphics[width = 0.5 \textwidth]{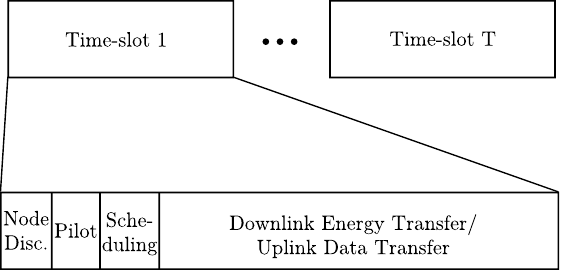}
\caption{A time slot structure. At the the beginning of each time slot, a small portion of time is devoted to node discovery, channel estimation through pilot transmission and scheduling. In the rest of the time slot the Reader transmits energy, and the BNs transmit data through backscattering. }
\label{fig:time}
\end{figure}

\subsection{Downlink Energy Transfer}
The Reader continuously transmits a carrier $c(t)$, which is the source of the energy required for uplink data transmission. Transmission beamforming technique is adopted in the downlink to increase the energy transmission efficiency. In time slot $t$, $f_m(t) \in \mathbb{C}, m= 1,\ldots, M$, denotes the gain of the $m$-th energy transmit path of the Reader, and $\bm{f}(t) \triangleq \big(f_1(t), \ldots, f_M(t)\big)^T \in \mathbb{C}^M$ denotes the transmit beamforming vector. The carrier $c(t)$ has unit power and $\lVert\bm{f}(t)\lVert^2 \leq P$ determines the transmission power of the Reader, with $P$ denoting the maximum transmission power. Considering the transmit beamforming vector and the channel gains, the received signal in $\text{BN}_n$ is given by
\begin{align}\label{eq:eDef}
e_n(t) = c(t)\bm{h}_n^T (t)\bm{f}(t) + v_n(t),
\end{align}
where $v_n(t)$ is the circular additive white Gaussian noise (AWGN) with variance $\sigma^2_v$. 

\subsection{Uplink Data Transfer }
The BNs modulate the backscattered signal by changing the impedance connected to the BN's antenna. In time slot $t$, $x_n(t)\in \mathbb{C}$, with $\operatorname{E}\{|x_n(t)|^2\} = 1$, and $\alpha_n(t) \in [0\; \alpha_\text{max}]$ denote the modulated signal and the reflection coefficient  of $\text{BN}_n$, respectively. Here, $\alpha_\text{max}$ denotes the maximum reflections coefficient of the BNs. The Reader receives  the sum of the backscattered signals from all BNs, that is,
\begin{align}\label{eq:RSig}
z(t) \triangleq \sum_{n\in\mathcal{N}} z_n(t) + \bm{w}(t),
\end{align}
where, $z_n(t)$ is the received signal from  $\text{BN}_n$, i.e.,
\begin{align}\label{eq:zDef}
z_n(t) \triangleq \alpha_n(t)e_n(t)x_n(t),
\end{align}
and $\bm{w}(t)\in \mathbb{C}^N$ denotes the AWGN with covariance matrix $\sigma_w^2I_M$.

The Reader separates the signal received from different BNs by multiplying the receive beamforming vector corresponding to each user with the received signal $z(t)$. Let, $\bm{g}_n(t) \triangleq \big(g_n^1(t), \ldots, g_n^M(t)\big)^T\in \mathbb{C}^M$, with $\lVert\bm{g}(t)\rVert^2 = 1$, denote the receive beamforming vector for $\text{BN}_n$, where $g_n^m(t) \in \mathbb{C},\;\forall m$, is the gain of the $m$ receive path.  The detected signal of  $\text{BN}_n$ is written as
\begin{align}
r_n(t) \triangleq \bm{g}_n^H(t)z_n(t) + \sum_{\tilde{n}\in \mathcal{N}/n}\bm{g}_n^H(t)z_{\tilde{n}}(t) + \bm{g}_n^H(t)\bm{w}(t).
\end{align}
Considering the definitions of $z_n(t)$ in \eqref{eq:zDef} and $e_n(t)$ in \eqref{eq:eDef},  the signal to interference plus noise ratio (SINR) of the signal of $\text{BN}_n$ at the decoder, $\mathcal{S}_n(t)$, is given by
\begin{align}\label{eq:SINR}
\mathcal{S}_n(t) \triangleq  \frac{\left | \alpha_n(t)  \bm{g}_n^H(t) \bm{h}_n(t) \bm{h}_n^T(t) \bm{f}(t)\right |^2 }{\sigma^2_w + \sum_{\tilde{n}\in \mathcal{N}/n} \left | \alpha_{\tilde{n}}(t)  \bm{g}_n^H(t) \bm{h}_{\tilde{n}}(t) \bm{h}_{\tilde{n}}^T(t) \bm{f}(t)\right |^2}.
\end{align}
Note that in practical systems the backscattered noise $v_n(t)h_n(t)$  is negligible compared to $w_n(t)$ due to channel attenuation \cite{Mishra_Min, Mishra_Sum, Lyu}, and hence it is neglected in \eqref{eq:SINR}. Moreover, the interference due to the unmodulated carrier leakage from the transmitter can be efficiently surpassed  \cite{Li2018, Chen, Mishra_Min, Mishra_Sum, Villame2010, Hao2018} and, accordingly, it is neglected in the SINR formulation \eqref{eq:SINR}.  Considering the SINR of the received backscattered signal and sufficiently long codewords, the data transmission rate of $\text{BN}_n$ in time slot $t$ is obtained as
\begin{align}\label{eq:Rn}
R_n(t) =W \log\big(1+\mathcal{S}_n(t)\big),
\end{align}
where $W$ is the channel bandwidth (for performance analysis in the cases with short packets, see Section \ref{sec:finite}).
\subsection{ Data Admission}
The BNs store the incoming data, e.g., the sensed data from environment, in their buffers.   There is a traffic shaping filter at the input of  each buffer in the BNs which limits the number of input bits in  time slot $t$  by $D_n(t) \in [0\;\;D_\text{max}],\;\forall n$, with $D_\text{max}$ being a constant.  Accordingly, the number of  stored data bits in the buffer of $\text{BN}_n$ in time slot $t$, denoted by $Q_n(t)$, evolves as
\begin{align}\label{eq:QueueEvolve}
Q_n(t+1)  = \big [Q_n(t) - R_n(t)\big ]^+ +D_n(t).
\end{align}
\subsection{Network Controller}
There is a  network controller at the Reader, which has access to the channel state information (CSI), $\bm{h}_n(t)\bm{h}^T_n(t)$, and the level of the stored data in the buffers,  $\bm{Q}(t) \triangleq \big(Q_1(t), \ldots, Q_N(t) \big)$. In each time slot, the controller determines the transmit beamforming vector  $\bm{f}(t)$ and the receive beamforming matrix $\bm{G}(t) \triangleq \big(\bm{g}_1(t),\ldots,\bm{g}_N(t)\big) \in \mathbb{C}^{M\times N} $. Moreover, the controller determines the reflection coefficient vector, $\bm{\alpha}(t) \triangleq \big(\alpha_1(t)\ldots, \alpha_N(t)\big)$, and the data admission vector, $\bm{D}(t) \triangleq \big(D_1(t),\ldots,D_N(t)\big)$. 
We design the network control algorithm to maximize a throughput utility function $\mathcal{U}(t)$ while keeping the data buffers stable. We consider three different throughput utility functions, including sum throughput utility 
 \begin{align}\label{eq:sumThDef}
\mathcal{U}(t) = \mathcal{U}_s\big(\bm{D}(t)\big) \triangleq \sum_{n \in \mathcal{N}} D_n(t),
\end{align}
proportional throughput utility
 \begin{align}\label{eq:propThDef}
\mathcal{U}(t) = \mathcal{U}_p\big(\bm{D}(t)\big) \triangleq \sum_{n \in \mathcal{N}} \log\big(1+D_n(t)\big),
\end{align}
and common throughput utility 
 \begin{align}\label{eq:comThDef}
\mathcal{U}(t) = \mathcal{U}_c\big(\bm{D}(t)\big) \triangleq\min_{n\in \mathcal{N}}\big\{D_n(t)\big\}.
\end{align}
The three introduced utility functions result in different  throughput and fairness tradeoffs. While the sum throughput utility is the most greedy function in terms of throughput, the common throughput utility is the most fair utility.  The proportional throughput utility, on the other hand, balances the throughput and fairness. Interestingly, each of these utility functions is useful in certain applications. 
\subsection{Problem Formulation}
Let $\mathcal{U}^\star$ denote the maximum utility achieved by the optimal control policy among all policies that stabilize the buffers.  In this way, considering the constraints on the transmission power, reflection coefficient and data admission, the network controller with maximum utility can be formulated as the solution of
\begin{maxi!}
{\substack{\bm{f}(t), \bm{G}(t), \bm{\alpha}(t), \bm{D}(t)  }} { \lim_{T \rightarrow \infty}\frac{1}{T}\sum_{t=0}^{T-1}\operatorname{E}\left\{\mathcal{U}(t)\right\}} {\label{prob:mainProbDef}}{ \mathcal{U}^\star = }
\addConstraint{ Q_n(t) < \infty, \; \forall n \label{eq:stableConstraint}}
\addConstraint{\lVert \bm{f}(t)\lVert^2 \leq P\label{eq:powerConstraint}}
\addConstraint{ \alpha_n(t) \in [0\;\;\alpha_\text{max}],\;\forall n\label{eq:alphaConstraint}}
\addConstraint{ D_n(t) \in [0\;\; D_\text{max}],\;\forall n, \label{eq:DConstraint}}
\end{maxi!}
where the expectation in \eqref{prob:mainProbDef} is with respect to the channel randomness. Constraint \eqref{eq:stableConstraint} ensures the stability of the buffers.  Constraints \eqref{eq:powerConstraint} and \eqref{eq:alphaConstraint}  are limitations on transmission power and reflection coefficients, respectively. Moreover, Constraint \eqref{eq:DConstraint} determines the maximum number of  data bits that  can be admitted to the buffers in each time slot. 

Problem \eqref{prob:mainProbDef} is a stochastic utility optimization problem. This problem can be tackled by dynamic programing (DP) methods. However, DP methods require the statistical knowledge of the channel state process, which may not be available. Here, we use the MDPP method \cite[Chapter 4]{Neely} to propose a solution for Problem \eqref{prob:mainProbDef}.  MDPP is a general framework for optimizing  time averages, with possibly  time average constraints (see, e.g., \cite{Huang2012,Neely2006,Movahednasab2020,MovahednasabICC,RRezaei2020,RezaeiPIMRC,RezaeiGlobe,Hadi2019} for different applications of the MDPP method).  Using this framework, a problem with a time average objective function  is reduced to a sub-problem which should be solved in each time slot.
 Accordingly, following the MDPP approach,  we formulate the optimal $\bm{f}(t), \bm{G}(t), \bm{\alpha}(t)$ and  $\bm{D}(t)$ in each time slot  as the solution of an optimization problem  with parameters $\bm{Q}(t)$ and $\bm{h}_n(t)\bm{h}_n^T(t)$. The formulated problem  is non-convex and does not have a trivial solution. However, we use quadratic and Lagrangian dual transforms \cite{Frac1,Frac2} to propose an iterative algorithm for finding the optimal control variables. 

\section{The Proposed Control Policy}\label{sec:propPolicy}
We follow the MDPP method to propose an online control policy that solves Problem \eqref{prob:mainProbDef}.  In summary, we follow these steps:
\begin{enumerate}
\item We define the Lyapunov function 
\begin{align}
L(t) \triangleq \frac{1}{2}\sum_{n\in \mathcal{N}} Q_n^2(t),
\end{align}
which is a scalar measure of the stored data in all buffers.

\item  We define the drift-plus-penalty (DPP) function 
\begin{align}\label{eq:driftpluspenalty}
\Delta_p\big( L(t)\big) \triangleq \operatorname{E}\Big\{ L(t+1) - L(t) \Big|\bm{Q}(t)  \Big\} - V\operatorname{E}\Big\{ \mathcal{U}(t) \Big\},
\end{align}
where $V>0$ is a control parameter. The first term in \eqref{eq:driftpluspenalty} shows the drift of the Lyapunov function in successive time slots. Positive or negative drift values indicate that the stored data in the buffers have increased or decreased, respectively. Moreover, the second term is a penalty which increases as the utility decreases.
Intuitively, we  expect that under a control policy, which minimizes  \eqref{eq:driftpluspenalty} in each time slot, the buffers will be stable and also the utility will be maximized. 
\item We introduce an upper bound for  the DPP function in Lemma \ref{lem:upperbound}.  Using the Lyapunov optimization  theorem \cite[Theorem 4.2]{Neely},   we  show in Theorem \ref{th:optimal} that under a policy which minimizes the developed upper-bound, the utility is maximized and the level of the stored data in the buffers is upper-bounded. 
\item To find the optimal control policy, we need to minimize the derived upper-bound in Lemma \ref{lem:upperbound}, which includes a non-convex function of $\bm{f}(t), \bm{G}(t)$ and $\bm{\alpha}(t)$. We use  Lagrangian dual  and quadratic transfers \cite{Frac1, Frac2}, and  propose an iterative algorithm for finding the optimal $\bm{f}(t), \bm{G}(t)$ and $\bm{\alpha}(t)$. We also derive the optimal data admission policy $\bm{D}(t)$. 
\end{enumerate}
The details of the analysis are explained as follows. We first introduce an upper-bound for $\Delta_p\big( L(t)\big)$ in Lemma \ref{lem:upperbound}. 
\begin{lemma}\label{lem:upperbound}
For the DPP function  \eqref{eq:driftpluspenalty}, we have
\begin{align}\label{eq:upperbound}
\begin{split}
 \Delta_p\big( L(t)\big) &\leq \Delta_u(t) \triangleq B - \sum_{n \in\mathcal{N}}\operatorname{E}\bigg\{Q_n(t)R_n(t)\Big|\bm{Q}(t)\bigg\}+\\
&\qquad\qquad\sum_{n \in\mathcal{N}}\operatorname{E}\bigg\{Q_n(t)D_n(t)\Big|\bm{Q}(t)\bigg\}  -V\operatorname{E}\bigg\{\mathcal{U}(t)\Big |\bm{Q}(t)\bigg\},
\end{split}
\end{align}
where 
\begin{align}
B \triangleq \frac{N}{2}\Big(D_\text{max} + R_\text{max}\Big).
\end{align}
Here, $R_\text{max} > 0$ is a sufficiently large constant, such that we always have $R_n(t)\leq R_\text{max}$.
\end{lemma}
\begin{proof}
See Appendix \ref{proof:upperbound}.
\end{proof}
The upper-bound function $\Delta_u(t)$ in \eqref{eq:upperbound} is the starting point for deriving the optimal policy. We formulate the optimization problem
\begin{mini!}
{\bm{f}(t), \bm{G}(t), \bm{\alpha}(t), \bm{D}(t)}{\Delta_u(t)}{\label{prob:upperBound}}{}
\addConstraint{\eqref{eq:powerConstraint}, \eqref{eq:alphaConstraint}, \eqref{eq:DConstraint}},
\end{mini!}
that, given $\bm{Q}(t)$ and $\bm{h}_n(t)\bm{h}_n^T(t)$ in each time slot, finds the values of $\bm{f}(t), \bm{G}(t), \bm{\alpha}(t)$ and $\bm{D}(t)$  minimizing $\Delta_u(t)$. In Theorem \ref{th:optimal}, we show that under a policy which solves Problem \eqref{prob:upperBound} the level of the stored bits in the buffers are bounded, and we can push the utility arbitrarily close to  $\mathcal{U}^\star$.
\begin{theorem}\label{th:optimal}
Suppose that $\bm{f}(t), \bm{G}(t), \bm{\alpha}(t)$ and $\bm{D}(t)$ in each time slot are determined according to the solution of Problem \eqref{prob:upperBound}. Then, we have
\begin{enumerate}
\item With utility functions \eqref{eq:sumThDef}, \eqref{eq:propThDef} and \eqref{eq:comThDef}, the level of the stored data in the buffers are upper-bounded by
\begin{align}\label{eq:buffUpper}
Q_n(t) \leq V + D_\text{max}, \;\; \forall n.
\end{align}
\item The average utility satisfies
\begin{align}\label{eq:Optimality}
\lim_{T \rightarrow \infty}\frac{1}{T}\sum_{t=0}^{T-1}\operatorname{E}\left\{\mathcal{U}(t)\right\} \geq \mathcal{U}^\star - \frac{B}{V}.
\end{align}
\end{enumerate}
\end{theorem}
\begin{proof}
See Appendix \ref{proof:theoremOpt}.
\end{proof}
According to \eqref{eq:buffUpper} in Theorem \ref{th:optimal}, if we adopt the solution of Problem \eqref{prob:upperBound} in each time slot, the level of the stored bits in the buffer will not exceed $ V + D_\text{max}$. Hence the buffers remain stable. Moreover, the performance bounds in \eqref{eq:buffUpper} and  \eqref{eq:Optimality} introduce  a tradeoff between the optimality gap of the utility and the size of the buffers. According to this tradeoff, while the utility optimality gap is within $\mathcal{O}(\frac{1}{V})$ the buffer size increases linearly with $V$.

We propose a solution for Problem \eqref{prob:upperBound} as follows. Considering  the $\Delta_u(t)$ function in \eqref{eq:upperbound}, Problem \eqref{prob:upperBound} can be separated into two independent problems. The first expectation in \eqref{eq:upperbound} is a function of $\bm{f}(t), \bm{G}(t)$ and 
$\bm{\alpha}(t)$, while the second and the third expectations are only functions of  $\bm{D}(t)$. Accordingly, we reformulate Problem  \eqref{prob:upperBound} into two sub-problems, including the link scheduling problem 
\begin{maxi!}
{\bm{f}(t), \bm{G}(t),\bm{\alpha}(t)}{ \sum_{n\in \mathcal{N}} Q_n(t) R_n(t)\label{prob:LinkScheduleObj}}{\label{prob:LinkSchedule}}{}
\addConstraint{\eqref{eq:powerConstraint}, \eqref{eq:alphaConstraint},}
\end{maxi!}
and the data admission problem
\begin{mini!} 
{\bm{D}(t)}{\sum_{n \in \mathcal{N}}Q_n(t) D_n(t) -V \mathcal{U}(t) }{\label{prob:DataAdmission}}{}
\addConstraint{\eqref{eq:DConstraint}.}
\end{mini!}
Note that we have removed the expectations  in Problems \eqref{prob:LinkSchedule} and \eqref{prob:DataAdmission} and we opportunistically minimize the expectations for realizations of the channel state. The link scheduling problem in  \eqref{prob:LinkSchedule} is a non-convex problem because of the product and ratio terms in  $R_n(t)$, i.e., the rate terms \eqref{eq:Rn}, and it does not have a trivial solution. However, the data admission problem in  \eqref{prob:DataAdmission} can be solved in closed-form for different utility functions in \eqref{eq:sumThDef}, \eqref{eq:propThDef} and \eqref{eq:comThDef}. 
\subsection{Data Link Scheduling Problem}
Problem \eqref{prob:LinkSchedule} includes maximizing a weighted sum of the BNs' data transmission rates. The objective function in \eqref{prob:LinkScheduleObj} contains multiple fractional terms, i.e., SINRs, which makes the problem NP-hard \cite{NP-hard}. However, we propose an iterative algorithm  which successively optimizes the variables.  To this end, we first find the optimal value for each  control parameter, $\bm{f}(t), \bm{G}(t)$ and $\bm{\alpha}(t)$, while the other two parameters are fixed. Finding the optimal value of the receive beamforming vector $\bm{g}_n(t),\;\forall n,$ while the transmit beamforming vector $\bm{f}(t)$ and the reflection coefficients $\bm{\alpha}(t)$ are fixed, is straightforward. This is because   the receive beamforming vector for $\text{BN}_n$, $\bm{g}_n(t)$, can be optimized independently 
through maximizing the SINR of $\text{BN}_n$, $\mathcal{S}_n(t)$. Moreover, $\mathcal{S}_n(t)$ can be formulated  as a generalized Rayleigh quotient, that is,
\begin{align}
\mathcal{S}_n(t) =  \frac{ \bm{g}_n^H(t)\bm{\zeta}_n(t)\bm{\zeta}^H_n(t)\bm{g}_n(t) }{\bm{g}_n^H(t)\bigg(\sigma^2_w\bm{I}_M + \sum_{\tilde{n}\in \mathcal{N}/n} \bm{\zeta}_{\tilde{n}}(t)\bm{\zeta}^H_{\tilde{n}}(t) \bigg)\bm{g}_n(t)},
\end{align}
where
\begin{align}
\bm{\zeta}_n(t)  = \alpha_n(t)\bm{h}_n(t)\bm{h}_n^T(t)\bm{f}(t).
\end{align}
Writing the stationarity Karush-Kuhn-Tucker conditions, \cite[Lemma 3.14]{Bjornson2013} shows that the generalized Rayleigh quotient is maximized by
\begin{align}\label{eq:optg}
\bm{g}_n^\star(t) = \frac{\bigg( \bm{I}_M + {1\over \sigma_w^2} \sum_{n \in \mathcal{N}} \bm{\zeta}_n(t)\bm{\zeta}^H_n(t)  \bigg)^{-1} \bm{\zeta}_n(t)}{\bigg\lVert  \bigg( \bm{I}_M + {1\over \sigma_w^2} \sum_{n \in \mathcal{N}} \bm{\zeta}_n(t)\bm{\zeta}^H_n(t)  \bigg)^{-1} \bm{\zeta}_n(t) \bigg\rVert }.
\end{align}

However, finding the closed-form optimal values of $\bm{f}(t)$ and $\bm{\alpha}(t)$ is more difficult, since they are coupled through the SINR terms of all BNs. We use the Lagrangian dual transform \cite{Frac2} and the quadratic transform \cite{Frac1} to reformulate Problem \eqref{prob:LinkSchedule} and facilitate finding the closed-form optimal values of $\bm{f}(t)$ and $\bm{\alpha}(t)$.  Lagrangian dual transform, introduced in Lemma \ref{prop:Lagrange},  converts the link scheduling problem to a problem of maximizing  the sum of ratios.
\begin{lemma}\label{prop:Lagrange}
The link scheduling problem \eqref{prob:LinkSchedule} is equivalent to 
\begin{maxi!}
{\bm{f}(t),\bm{G}(t), \bm{\alpha}(t), \bm{\gamma}}{\tilde{R}(t, \bm{\gamma}) }{\label{prob:lagrange}}{}
\addConstraint{\eqref{eq:powerConstraint}, \eqref{eq:alphaConstraint},}
\end{maxi!}
in the sense that Problem \eqref{prob:lagrange} leads to the same solution and maximum value as Problem \eqref{prob:LinkSchedule}. Here, $\bm{\gamma} = (\gamma_1, \ldots, \gamma_N)^T \in \mathcal{R}^N$ is an auxiliary variable and the objective function $\tilde{R}(t, \bm{\gamma})$ is defined as
\begin{align}\label{eq:Rgamma}
\tilde{R}(t, \bm{\gamma}) \triangleq \sum_{n  \in \mathcal{N}} Q_n(t)\log(1+\gamma_n)-\sum_{n \in \mathcal{N}}Q_n(t)\gamma_n +\sum_{n  \in \mathcal{N}}\frac{Q_n(t)(1+\gamma_n)|\alpha_n(t) \bm{\beta}_{n,n}^T(t) \bm{f}(t)|^2}{\sigma^2_w+\sum_{\tilde{n} \in \mathcal{N}} |\alpha_n(t)\bm{\beta}^T_{n,\tilde{n}}(t)\bm{f}(t)|^2},
\end{align}
with $\bm{\beta}^T_{n,\tilde{n}}(t)= \bm{g}_n^H(t)\bm{h}_{\tilde{n}}(t)\bm{h}_{\tilde{n}}^T(t)$.
\end{lemma}
\begin{proof}
$\tilde{R}(t, \bm{\gamma})$ is a concave and differentiable function. Hence, the stationary point 
\begin{align}\label{eq:optGamma}
\gamma_n^\star = \frac{|\alpha_n(t) \bm{\beta}_{n,n}^T(t) \bm{f}(t)|^2}{\sigma^2_w+\sum_{\tilde{n}\in \mathcal{N}/n} |\alpha_n(t)\bm{\beta}^T_{n,\tilde{n}}(t)\bm{f}(t)|^2}
\end{align}
is the optimal solution which maximizes $\tilde{R}(t, \bm{\gamma})$. Substituting $\gamma_n^\star,\;\forall n$, in \eqref{eq:Rgamma} recovers $\sum_{n\in \mathcal{N}}Q_n(t)R(t)$ which establishes  the equivalence. 
\end{proof}
In the transformed objective function \eqref{eq:Rgamma} there is no logarithm function, but the last term in  \eqref{eq:Rgamma}  is still in fractional form.  We use the quadratic transform to convert the fractions  to an equivalent summation without fractions. 
\begin{lemma}
The optimization problem \eqref{prob:lagrange} is equivalent to
\begin{maxi!}
{\bm{f}(t),\bm{G}(t), \bm{\alpha}(t), \bm{\gamma}, \bm{y}}{\vardbtilde{R}(t, \bm{\gamma}, \bm{y}) }{\label{prob:y}}{}
\addConstraint{\eqref{eq:powerConstraint}, \eqref{eq:alphaConstraint},}
\end{maxi!}
in the sense that problem \eqref{prob:y} leads to the same solution and maximum value as problem \eqref{prob:lagrange}. Here, $\bm{y} = (y_1,\ldots,y_n)^T\in \mathbb{C}^N$ is an auxiliary variable and the objective function $\vardbtilde{R}(t, \bm{\gamma},\bm{y})$ is defined as
\begin{align}\label{eq:Ry}
\begin{split}
\vardbtilde{R}(t,\bm{\gamma}, \bm{y}) \triangleq \sum_{n\in \mathcal{N} } &Q_n(t)\log(1+\gamma_n)-\sum_{n\in \mathcal{N}}Q_n(t)\gamma_n +\\
&\sum_{n\in \mathcal{N}} 2\sqrt{Q_n(t)(1+\gamma_n)} \operatorname{Re}\bigg\{y_n\alpha_n(t)\bm{\beta}_{n,n}^T(t)\bm{f}(t)\bigg\}-\\
&\sum_{n\in \mathcal{N}}|y_n|^2\bigg(\sigma^2_w+\sum_{\tilde{n}\in \mathcal{N}}|\alpha_{\tilde{n}}(t)\bm{\beta}_{n,\tilde{n}}^T(t) \bm{f}(t)|^2\bigg)
\end{split}
\end{align}
\end{lemma}
\begin{proof}
Taking complex derivative $\partial \vardbtilde{R} / \partial y_n$ and solving $\partial \vardbtilde{R}/\partial y_n = 0$, the optimal values of $y_n$ are obtained by solving 
\begin{align}
\sqrt{Q_n(t)(1+\gamma_n)}   \alpha_n(t) \bm{\beta}_{n,n}^H(t) \bm{f}^\ast(t)-y_n \left( \sigma^2_w+\sum_{\tilde{n} \in \mathcal{N}}^N |\alpha_n(t)\bm{\beta}^T_{n,\tilde{n}}(t)\bm{f}(t)|^2 \right) =0,
\end{align}
which leads to
\begin{align}\label{eq:optY}
y_n^\star = \frac{\sqrt{Q_n(t)(1+\gamma_n)}   \alpha_n(t) \bm{\beta}_{n,n}^H(t) \bm{f}^\ast(t)}{\sigma^2_w+\sum_{\tilde{n} \in \mathcal{N}}^N |\alpha_n(t)\bm{\beta}^T_{n,\tilde{n}}(t)\bm{f}(t)|^2}.
\end{align}
Substituting \eqref{eq:optY} in \eqref{eq:Ry} recovers the $\tilde{R}(t, \bm{\gamma})$. Hence, the equivalence is concluded.
\end{proof}
 The reformulated Problem \eqref{prob:y} enables us to find the closed-form optimal value of the transmit beamforming vector, $\bm{f}(t)$, or the reflection coefficients, $\bm{\alpha}(t)$, when the  other one  is fixed. 
Specifically, considering the power constraint \eqref{eq:powerConstraint}, we find the optimal   $\bm{f}(t)$,
 by introducing the dual variable $\eta \geq 0$  and solving $\partial \vardbtilde{R}/ \partial \bm{f}+ \eta\partial\big(\lVert\bm{f}(t)\lVert^2 - P\big)/\partial \bm{f}= 0$. Accordingly, we obtain
\begin{align}\label{eq:optf}
\bm{f}^\star(t) = \bigg(\sum_{n\in \mathcal{N}}|y_n|^2 \sum_{\tilde{n} \in \mathcal{N}} |\alpha_{\tilde{n}}(t)|^2\beta^\ast_{n,\tilde{n}}(t)\beta_{n,\tilde{n}}^T(t) + \eta \bm{I}\bigg)^{-1} \sum_{n \in \mathcal{N}}\bigg( \sqrt{Q_n(t)(1+\gamma_n)}\beta_{n,n}^\ast(t) \alpha_n(t) y_n^\ast \bigg),
\end{align}
where $\eta$ satisfies
\begin{align}\label{eq:eta}
\eta = \min \Big\{ \eta \geq 0 : \lVert\bm{f}^\star(t)\lVert^2 \leq P  \Big\}.
\end{align}
The dual variable $\eta$ in \eqref{eq:eta} is determined by, e.g., bisection search.

Likewise, solving $\partial \vardbtilde{R}/\partial \alpha_n = 0$,  we obtain
\begin{align}\label{eq:optAlpha}
\alpha^\star_n(t) = \begin{cases}
0&\alpha^\text{st}_n(t) \leq0\\
\alpha^\text{st}_n(t) &0\leq\alpha^\text{st}_n(t)\leq\alpha_\text{max}\\
\alpha_\text{max}&\alpha_\text{max}\leq\alpha^\text{st}_n(t),
\end{cases}
\end{align}
where,
\begin{align}
\alpha_n^\text{st}(t) = \frac{\sum_{n\in\mathcal{N}} \sqrt{Q_n(t)(1+\gamma_n)} \mathbb{R}\Big\{y_n\bm{\beta}_{n,n}^T(t)\bm{f}(t)\Big\} }{ \sum_{\tilde{n} \in \mathcal{N}} |y_{\tilde{n}}|^2|\bm{\beta}_{\tilde{n},n}^T(t) \bm{f}(t)|^2}.
\end{align}
Having a closed-form solution for each variable, we propose an iterative algorithm, for updating $\bm{f}(t), \bm{g}_n(t)$ and $\bm{\alpha}(t)$. The Link scheduling algorithm is summarized in Algorithm \ref{alg:ED}. Having access to $Q_n(t)$ and $\bm{h}_n(t)\bm{h}_n^T$  the Reader runs  Algorithm \ref{alg:ED} at the beginning of each time slot. In Algorithm \ref{alg:ED}, the transmit beamforming vector is initialized by 
\begin{align}\label{eq:fInit}
\bm{f}(t) = \frac{ \sum_{n \in \mathcal{N} }Q_n(t) \bm{h}^\ast_n(t)}{ \lVert \sum_{n \in \mathcal{N} }Q_n(t) \bm{h}^\ast_n(t) \lVert} \sqrt{P}.
\end{align}
This  initial point for $\bm{f}(t)$ is intuitive, since in a network with a single BN the maximum transmission ratio (MRT) beamforming, $\bm{f}(t) = \bm{h}^\ast(t) / \lVert \bm{h}^\ast(t)\lVert$,  is optimal \cite{Bjornson2014}. Moreover,  the weights $Q_n(t)$ in \eqref{eq:fInit} are motivated by the fact that the nodes with more congested buffers need more power to achieve higher data transmission rate. The iterations in Algorithm \ref{alg:ED} are terminated if the improvement of the objective function \eqref{prob:LinkScheduleObj} is below a threshold, determined by a convergence threshold parameter $\epsilon$, or if the number of  iterations exceed some pre-defined value $\text{it}_\text{max}$. Theorem \ref{prop:conv} establishes the convergence of Algorithm \ref{alg:ED}.
\begin{theorem}\label{prop:conv}
Algorithm \ref{alg:ED} is guaranteed to converge. Moreover, the objective function \eqref{prob:LinkScheduleObj} is non-decreasing in each iteration. 
\end{theorem}
\begin{proof}
See Appendix \ref{proof:probConv}.
\end{proof}

The simulation results in Section \ref{sec:Numerical} show that, for a broad range of parameter settings, Algorithm \ref{alg:ED} converges with few numbers of iterations.

\begin{algorithm}[t]
\caption{ Energy and data transmission scheduling in  time slot $t$.}\label{alg:ED}
 \hspace*{\algorithmicindent} \textbf{Input:} $\bm{Q}(t), \bm{h}_n(t)\bm{h}_n^T \; \forall n, P$, convergence tolerance $ \epsilon$ and maximum iterations $\text{it}_\text{max}$.   \\
 \hspace*{\algorithmicindent} \textbf{Output:}  $\bm{f}(t), \bm{g}_n(t), \alpha_n(t) \; \forall n $.\\
\begin{algorithmic}[1]
\State $f(t) \leftarrow \frac{ \sum_{n \in \mathcal{N} }Q_n(t) \bm{h}^\ast_n(t)}{ \lVert \sum_{n \in \mathcal{N} }Q_n(t) \bm{h}^\ast_n(t) \lVert} \sqrt{P}$.  \Comment{Initialization}
\State $\alpha(t) \leftarrow  [\alpha_\text{max}, \ldots,\alpha_\text{max}]^T$.
\State Calculate    $\bm{g}_n(t), \; \forall n,$  according to   \eqref{eq:optg}.
\State $r_c \leftarrow 0$.
\State $\text{it} \leftarrow 0$
\Repeat \Comment{Main loop}
\State $r_p \leftarrow r_c$.
\State Update $\gamma_n(t),\;\forall n,$ according to \eqref{eq:optGamma}.
\State Update $y_n(t),\;\forall n,$ according to \eqref{eq:optY}.
\State Update $\bm{f}(t)$ according to \eqref{eq:optf}.
\State Update $\alpha_n(t),\;\forall n,$ according to \eqref{eq:optAlpha}.
\State Update $\bm{g}_n(t),\;\forall n,$ according to \eqref{eq:optg}.
\State Calculate $R_n(t),\;\forall n,$ according to \eqref{eq:Rn}.
\State $r_c \leftarrow \sum_{n\in \mathcal{N}} Q_n(t) R_n(t)$
\State $\text{it} \leftarrow \text{it} + 1$.
\Until{$|r_c - r_p| > \epsilon r_p \land \text{it}  \leq \text{it} _\text{max} $}
\end{algorithmic}
\end{algorithm}

\subsection{Data Admission Problem}
Considering different utility functions, we derive the corresponding admission policy through solving Problem \eqref{prob:DataAdmission}. 
\paragraph{Sum Throughput Utility}
 Problem \eqref{prob:DataAdmission} with the sum throughput  utility $\mathcal{U}(t) = U_s(\bm{D}(t))$ converts to the  linear problem
\begin{mini!} 
{\bm{D}(t)}{\sum_{n \in \mathcal{N}}(Q_n(t)-V) D_n(t) }{}{}
\addConstraint{D_n(t) < D_\text{max}, \;\;\forall n}
\end{mini!}
which is solved if 
\begin{align}\label{eq:sumThroughputPolicy}
D_n(t) = \begin{cases} D_\text{max}, & Q_n(t) \leq V\\
0, &\text{Otherwise}.
\end{cases}
\end{align}
The policy in \eqref{eq:sumThroughputPolicy} follows a greedy binary rule. Specifically, the BNs admit the maximum possible data in their buffers, whenever the stored bit level in the buffer is below $V$, while no data is admitted when the level reaches $V$.
\paragraph{Common Throughput Utility}
Adopting the common throughput  utility, $U_c(\bm{D}(t))$ in Problem \eqref{prob:DataAdmission}, we have 
\begin{mini!} 
{\bm{D}(t)}{\sum_{n \in \mathcal{N}} Q_n(t) D_n(t) - V \min_{n}\{D_n(t)\} \label{eq:objComm}}{\label{prob:dataAdmissionCommon}}{}
\addConstraint{D_n(t) \leq D_\text{max}, \;\;\forall n. \label{eq:ConstDmaxComm}}
\end{mini!}
The following Lemma establishes the structure of the solution of Problem \eqref{prob:dataAdmissionCommon}.
\begin{lemma}\label{prop:commonThStruct}
For maximizing the objective function \eqref{eq:objComm}, all BNs should admit equal amount of data. That is, 
\begin{align}\label{eq:optfeature}
D_1(t) = D_2(t) = \ldots= D_N(t) = D(t).
\end{align}
\end{lemma}
\begin{proof}
Considering an arbitrary data admission vector $\bm{D}(t)$ that satisfies \eqref{eq:ConstDmaxComm}, we construct a new vector $\tilde{\bm{D}}(t)$ such that $\tilde{D}_1(t)= \tilde{D}_2(t) = \ldots = \min_{n}\{D_n(t) \}$. Then, we  have
\begin{align}\label{eq:lemFeatureIn}
\sum_{n \in \mathcal{N}} Q_n(t) \tilde{D}_n(t) - V \min_{n}\{\tilde{D}_n(t)\} \leq \sum_{n \in \mathcal{N}} Q_n(t) D_n(t) - V \min_{n}\{D_n(t)\},
\end{align}
where equality holds only if $\bm{D}(t) = \tilde{\bm{D}}(t)$. Inequality \eqref{eq:lemFeatureIn} holds since  we have  $\min_{n}\{D_n(t)\} = \min_{n}\{\tilde{D}_n(t)\}$ and $  \tilde{D}_n(t) \leq D_n(t)$. Accordingly, the optimal $\bm{D}(t)$  follows \eqref{eq:optfeature}.
\end{proof}
Using Lemma \ref{prop:commonThStruct}, after adopting $\bm{D}(t) = (D(t),\ldots, D(t))^T$, we can rewrite Problem \eqref{prob:dataAdmissionCommon} as
\begin{mini!} 
{D(t)}{D(t)\sum_{n \in \mathcal{N}} (Q_n(t)  - V )\label{eq:objComm2}}{\label{prob:dataAdmissionCommon2}}{}
\addConstraint{D(t) \leq D_\text{max}.  \label{eq:ConstDmaxComm2}}
\end{mini!}
Thus, the optimal data admission rule is
\begin{align}\label{eq:OptAdmissionComm}
D_n(t) = \begin{cases} D_\text{max}, &\sum_{n\in \mathcal{N}} Q_n(t) \leq V\\
0 &\text{Otherwise}.
\end{cases}
\end{align}
Similar to \eqref{eq:sumThroughputPolicy}, the data admission policy in \eqref{eq:OptAdmissionComm} follows a binary rule. However, the admission decision in \eqref{eq:OptAdmissionComm} is common for all BNs. Accordingly, under the common throughput utility all BNs will reach the same throughput.  

\paragraph{Proportional Throughput Utility}
Substituting the proportional  throughput utility \eqref{eq:propThDef} in Problem \eqref{prob:DataAdmission}, we have 
\begin{mini!} 
{\bm{D}(t)}{\sum_{n \in \mathcal{N}} Q_n(t) D_n(t) - V\sum_{n \in \mathcal{N}}  \log\big(1+D_n(t)\big) \label{eq:objProp}}{\label{prob:dataAdmissionProp}}{}
\addConstraint{D_n(t) \leq D_\text{max}, \;\;\forall n, \label{eq:ConstDmaxProp}}
\end{mini!}
which is a convex and differentiable function with respect to $\bm{D}(t)$. Accordingly, comparing the stationary point of \eqref{eq:objProp} and the boundary point in \eqref{eq:ConstDmaxProp}, we find the optimal admission policy
\begin{align}\label{eq:optPropAdmission}
D_n(t) = \begin{cases} D_\text{max}, &    Q_n(t) \leq \frac{V}{1+D_\text{max}}\\
0, &Q_n(t) \geq V\\
\frac{ V}{Q_n(t)} -1, &\text{Otherwise}.
\end{cases}
\end{align}
The admission policy in \eqref{eq:optPropAdmission} is proportional to the inverse of the level of the stored data in the buffer. This, intuitively, means that the BN becomes less greedy to admit new data when the buffer size increases.

\section{Simulation Results}\label{sec:Numerical}

Here, we present the simulation results, and evaluate the performance of our proposed policy. In all figures, we consider the Rician fading model, that is,
\begin{align}
\bm{h}_n(t) = \sqrt{\beta_n}\Bigg(  \sqrt{\frac{K}{K+1}} \bm{h}^d_n(t)  +   \sqrt{\frac{1}{K+1}} \bm{h}^s_n(t)  \Bigg),
\end{align}
where, $\bm{h}^d_n(t)$ and $\bm{h}^s_n(t)$ denote the deterministic and scattered components of the channel, respectively. Moreover, $K$ is the Rician $K$-factor which determines the ratio between the Rician and the scattered components, and   $\beta_n$ represents the path loss factor. Note that $K=0$ and $K\to\infty$ represent the cases with Rayleigh fading and line-of-sight channels, respectively. We consider $\beta_n = d_n^{-\rho}\left({3\times10^8\over 4\pi f}\right)^2$, where $d_n$ is the distance between $\text{BN}_n$ and the Reader, $\rho$ is the path loss exponent and $f$ is the transmit frequency.
The entries of the scattered component vector $\bm{h}^s_n(t)$ are independent and zero-mean unit variance circularly symmetric complex Gaussian (CSCG) distributed random variables. {The deterministic components, $\bm{h}^d_n(t)$, are determined  according to a half wavelength separated uniform linear array setting as modeled in \cite[Eq. 2]{Zeng2015}.}
 Unless otherwise stated, the BNs are distributed randomly with a uniform distribution in a circular area around the Reader with average distance to Reader equal to $30$ m. The general simulation parameters are as summarized in Table \ref{table:simParam}.

\begin{table}
\centering
\caption{Summary of the simulation parameters.}\label{table:simParam}
\centering
\begin{tabular}{|c|c|}
\hline
Parameter & Value\\
\hline
Number of BNs, $N$ & 5\\
\hline
Number of Reader's antennas, $M$ & 5\\
\hline
Received noise power, $\sigma_w $ & $-110$ dBm\\
\hline
Reader's maximum transmission power, $P$ & $500$ mW\\
\hline
Bandwidth, $W$ & $5$ kHz\\
\hline
Maximum admitted data in each  time slot, $D_\text{max}$ & $30$ kbits\\
\hline
Convergence threshold in Algorithm \ref{alg:ED},  $\epsilon$ & 0.01\\
\hline
 Maximum iterations in Algorithm \ref{alg:ED}, $\text{it}_\text{max}$ & 100\\
\hline 
Maximum reflection coefficient, $\alpha_\text{max}$ & 0.8 \\
\hline
Simulation time, $T$ & $10^3$\\
\hline
Rician factor, $K$ & 1\\
\hline
Path loss exponent, $\rho$ & 3\\
\hline
Transmit frequency, $f$ & $915$ MHz\\
\hline
\end{tabular}
\end{table}

Considering different parameter settings, Fig. \ref{fig:iteration} studies the convergence of Algorithm \ref{alg:ED}. Specifically, in Figures \ref{fig:iteration_avgN}  and \ref{fig:iteration_avgM}, the average number of  iterations of Algorithm \ref{alg:ED} is plotted  versus the number of BNs $N$ and the number of Reader's antennas $M$, respectively. As seen in the figures, with different parameter settings, Algorithm \ref{alg:ED} converges with few iterations. However, the average number of iterations increases slightly as the number of  BNs or the number of Reader's antennas increase. 

In Fig.  \ref{fig:iteration_max}, we study the effect of the maximum number of iterations $\text{it}_\text{max}$, on the achievable  sum throughput utility $\bar{U}_s \triangleq {1 \over T} \sum_{t= 0}^{T-1} U_s(t)$, in a BCN with $N = \{5, 10\}$ and $M = \{5, 10\}$.  According to this figure, we achieve the maximum utility with a few iterations of Algorithm \ref{alg:ED} in each time slot. Moreover,  considering only one iteration, we see in Fig.  \ref{fig:iteration_max} that we can achieve almost $95\%$ of the maximum  utility. Accordingly, we can  reduce the scheduling time through reducing the maximum number of  iterations of Algorithm \ref{alg:ED} significantly, with marginal performance degradation.

\begin{figure}
\centering
\begin{subfigure}[b]{0.31\textwidth}
\includegraphics[width =1 \textwidth]{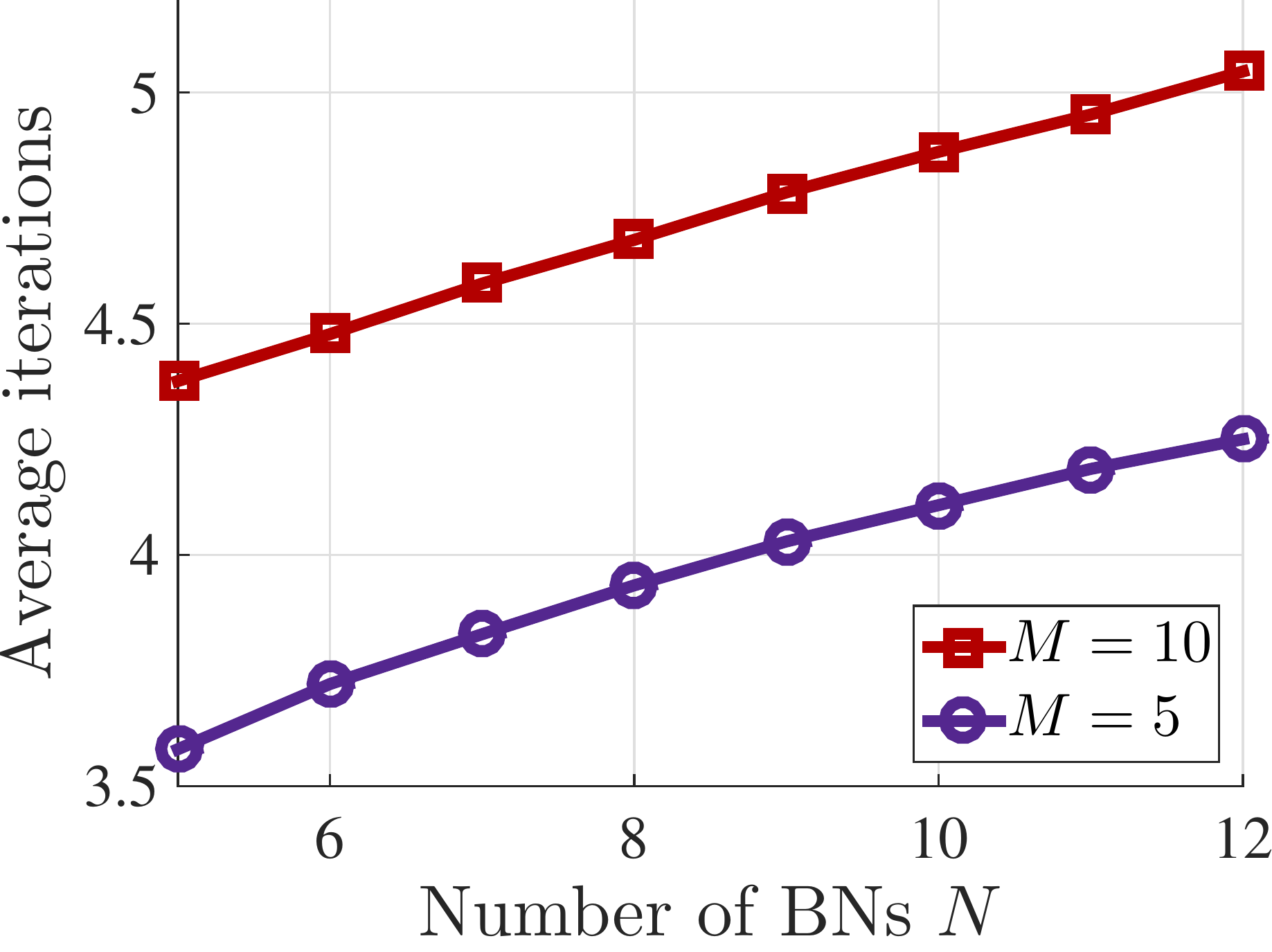}
\caption{  } \label{fig:iteration_avgN}
\end{subfigure}
~
\begin{subfigure}[b]{0.31\textwidth}
\includegraphics[width =1 \textwidth]{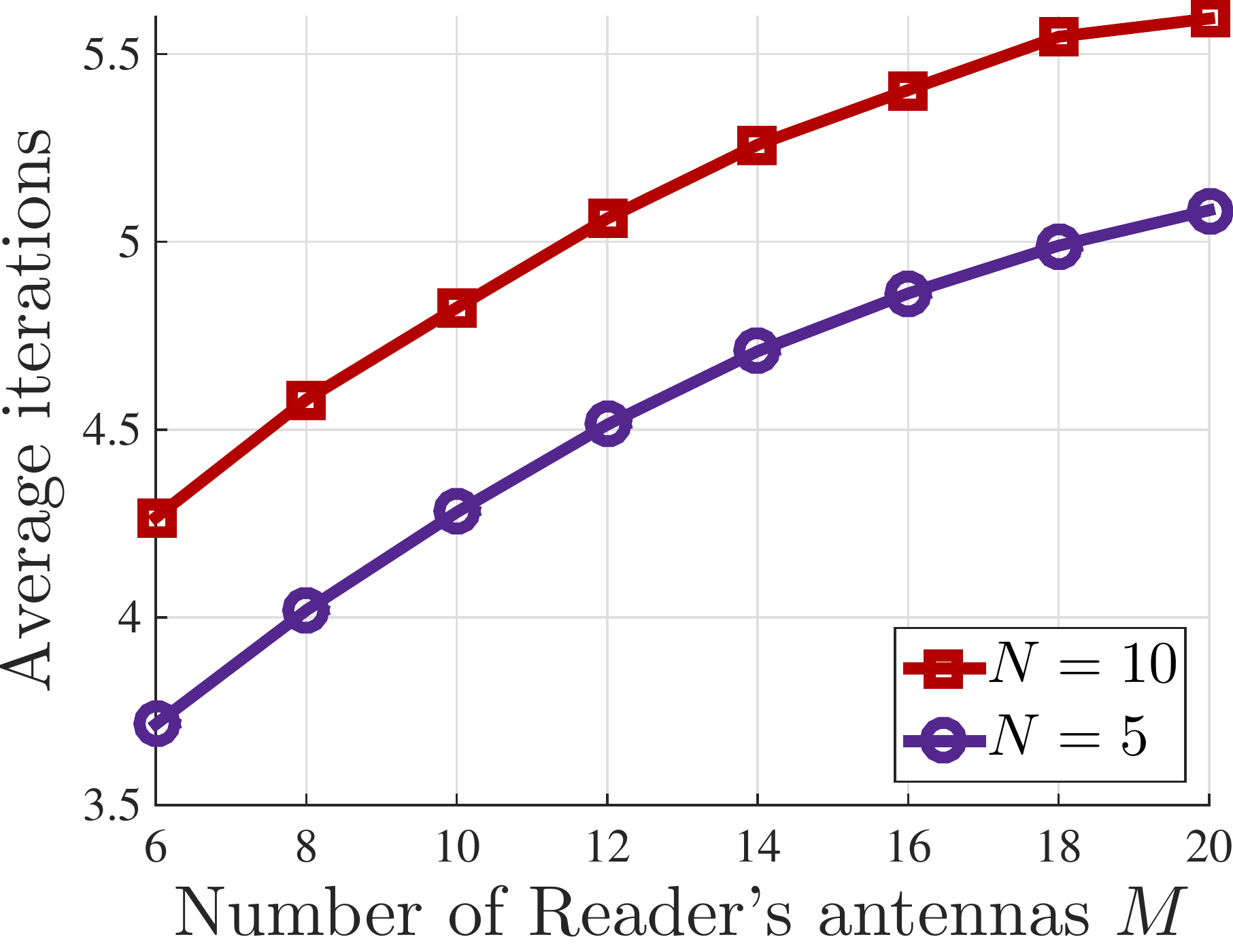}
\caption{  } \label{fig:iteration_avgM}
\end{subfigure}
~
\begin{subfigure}[b]{0.33\textwidth}
\includegraphics[width =1 \textwidth]{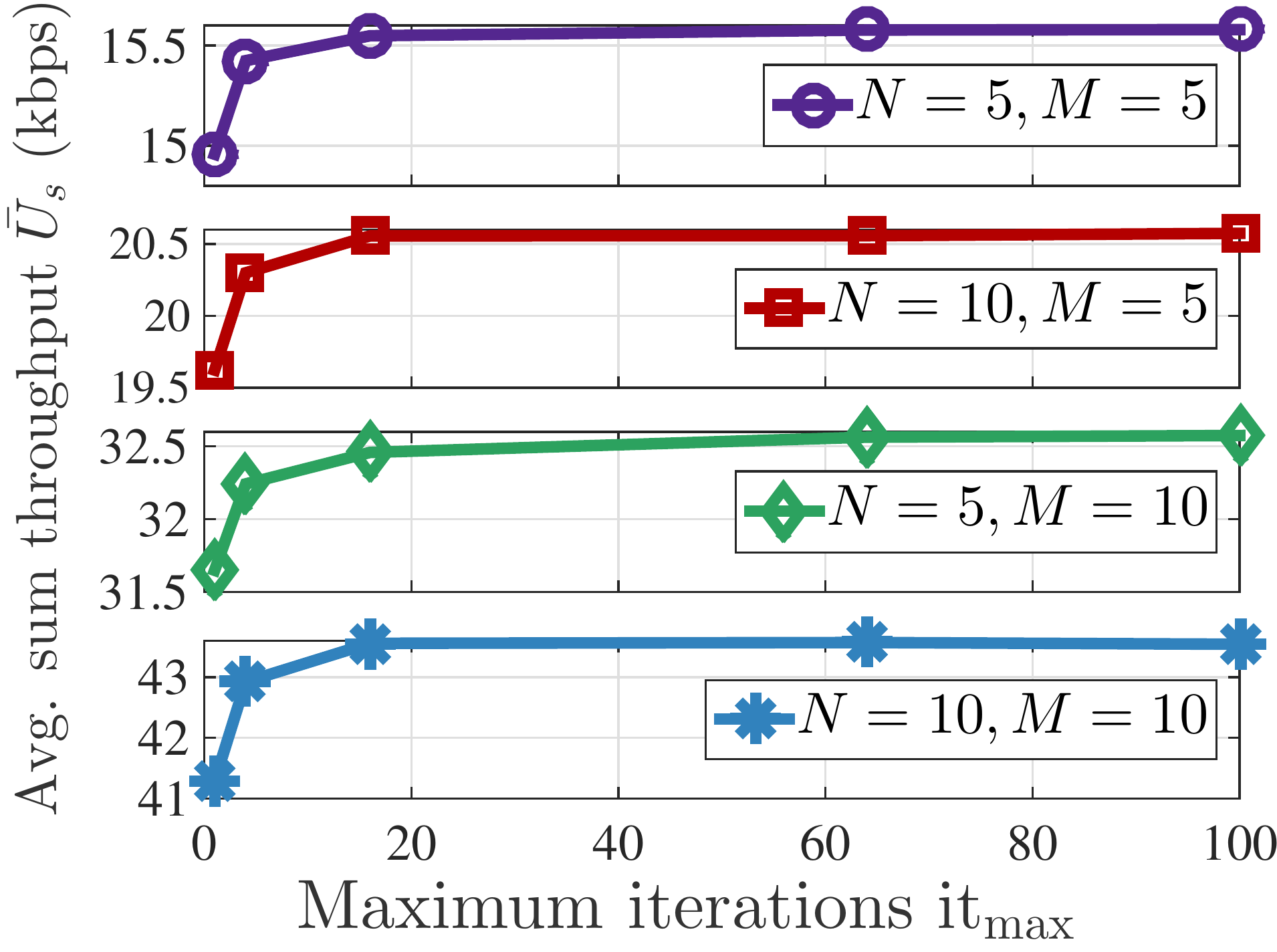}
\caption{ }\label{fig:iteration_max}
\end{subfigure}

\caption{ Convergence analysis of Algorithm \ref{alg:ED} with different parameter settings. Figures (a) and (b) show the average iterations versus $N$ and $M$, respectively. Figure (c) shows the average sum throughput utility versus the maximum number of  iterations $\text{it}_\text{max}$.  } \label{fig:iteration}
\end{figure}

Considering maximum transmission power $P = \{100, 800\}$ mW and different numbers of  Reader's antennas  $M = \{6, 8, 10\}$, Fig. \ref{fig:UvsQ} demonstrates the tradeoff between the average proportional throughput utility, $\bar{U}_p = {1\over T}\sum_{t = 0}^{T-1} U_p(t)$, and the average data level in the buffers.  Here, for given values of $P$ and $M$, the  utilities   are obtained under different values of $V$   between $10^7$ and $10^9$.  As seen  in Fig. \ref{fig:UvsQ},  the utility increases as the  buffers become more congested.  However, the utility saturates as the average data level in the buffers increases. This tradeoff between utility and the average data level of the buffers is in harmony with Theorem \ref{th:optimal}. That is, the optimality gap is inversely proportional to the maximum data level in the buffers.
Furthermore, considering different parameter settings in Fig. \ref{fig:UvsQ}, the tradeoff between utility and the data level in  buffers is almost insensitive to the number of  Reader's antennas $M$ or the transmission power $P$. That is, with different values of $P$ and $M$ the utility saturates when the data level in the buffers exceeds $50$ kbits.
 Finally, as expected, we see in Fig. \ref{fig:UvsQ} that the utility improves as    $M$ or $P$ increases. However, this improvement becomes less dominant as the utility  increases.

\begin{figure}
\centering
\includegraphics[width =0.6 \textwidth]{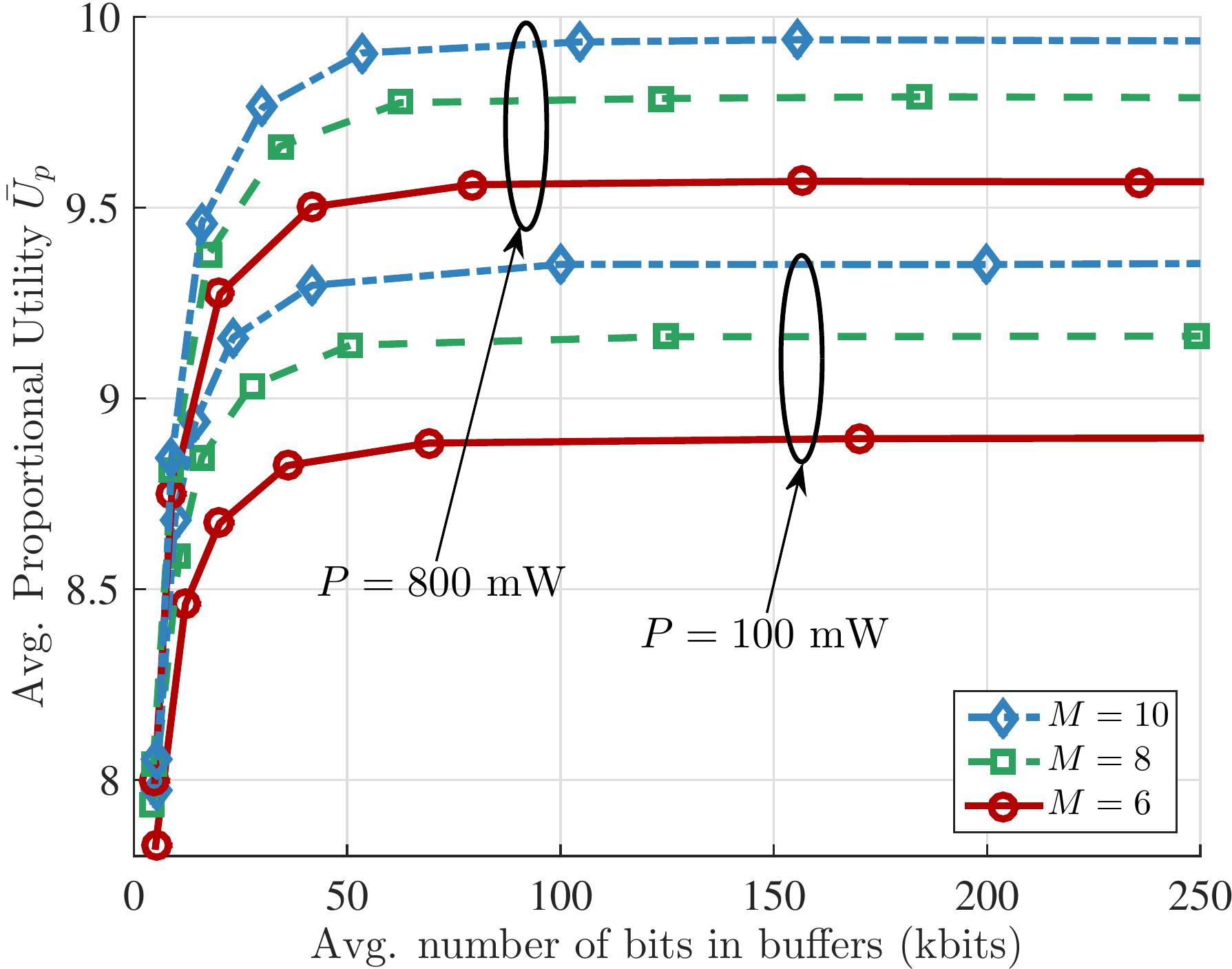}
\caption{ Average proportional throughput utility versus the average data level in the BNs' buffers for $P = \{100, 800\}$ mW and $M = \{6, 8, 10\}$.} 
\label{fig:UvsQ}
\end{figure}

Considering different values for the average distance of the BNs to the Reader, Fig. \ref{fig:commonVsMishra} shows the average common throughput utility, $\bar{U}_c = {1\over T}\sum_{t = 0}^{T-1} U_c(t)$, versus the number of  Reader's antennas $M$. Moreover, the figure compares our scheme with the proposed policy in \cite{Mishra_Min}, which maximizes the minimum achievable rate of all BNs in each time slot.  As can be seen, with the considered average BNs distances and the number of the Reader's antennas in Fig. \ref{fig:commonVsMishra}, the common throughput utility under our proposed policy improves on average $13\%$, compared to the policy in \cite{Mishra_Min}. The  reason for such improvement is adopting buffers in the BNs. The buffers allow the BNs to delay the transmission of the admitted data in the BNs. Accordingly, the link control policy in Algorithm \ref{alg:ED}  will optimally schedule data transmission for each BN in the best time slot and,  hence, optimizes the resources to maximize the average  utility. Furthermore, Fig. \ref{fig:commonVsMishra} shows the importance of using multi-antenna Readers in BCNs. According to this figure, considering different average BNs distances, the  utility improves more than $100\%$ when $M$ increases from $6$ to $20$. Also,  adopting more antennas in the  Reader, we can increase the coverage range. For example, 
consider the parameter settings of Fig. \ref{fig:commonVsMishra} and the common throughput utility of $1.2$ kbps. Then, by increasing the number of antennas from $6$ to $10$ the average supported distance increases from 24 m to 36 m. However, this relative improvement becomes less dominant as the number of antennas increases.

\begin{figure}
\centering
\includegraphics[width =0.6 \textwidth]{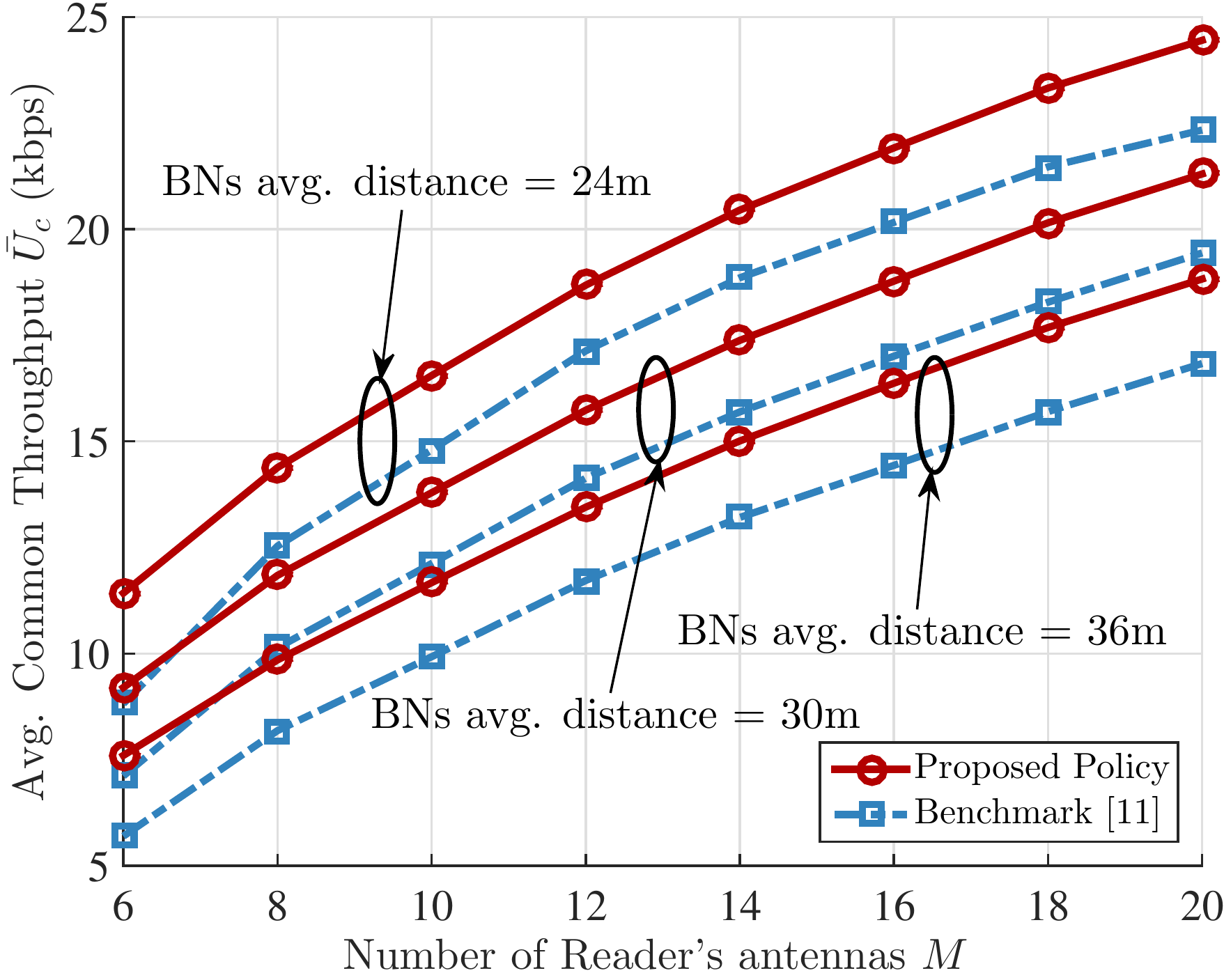}
\caption{ Common throughput utility versus the number of the Reader's antennas $M$, for the proposed algorithm and the benchmark in \cite{Mishra_Min}. BNs are distributed randomly in a circular area with average distance to the Reader  $\{24, 30, 36\}$ m and $V = 10^5$. }  \label{fig:commonVsMishra}
\end{figure}

In Fig. \ref{fig:sumVsN},  we compare the average sum throughput utility under our proposed policy with two other benchmarks. For the first benchmark, called  MRT, we adopt the transmit beamforming according to $\bm{f}(t) = \sum_{n\in \mathcal{N}} \bm{h}_n(t) / \lVert  \sum_{n\in \mathcal{N}} \bm{h}_n(t) \lVert$  and use the optimal receive beamforming in \eqref{eq:optg}. This specific choice of  $\bm{f}(t)$ is motivated by the fact that MRT beamforming is optimal for the single BN scenario \cite{Bjornson2014}. For the second benchmark, we consider the  joint transmit and receive beamforming design proposed in \cite{Mishra_Sum}. The BNs of the considered model in \cite{Mishra_Sum}  have no buffers. Hence, in  \cite{Mishra_Sum} the sum throughput utility is optimized in a time slot based framework. Considering $N \in \{5,\ldots, 12\}$ BNs distributed uniformly in an area with radius $\{50, 70\}$ m,  Fig. \ref{fig:sumVsN} demonstrates the average sum throughput utility versus the number of  BNs $N$. 

As  seen in Fig. \ref{fig:sumVsN}, the sum throughput utility increases almost linearly with the number of  BNs $N$. Moreover, according to this figure, the difference between the  utility under our proposed policy and the  policy in  \cite{Mishra_Sum} is  negligible. Also, both policies provide an average improvement of about $13\%$ and $24\%$ over the MRT  with radius equal to $70$ m and $50$ m, respectively. Accordingly, unlike the common throughput utility, with the sum throughput utility, adopting  buffers in the BNs will not increase the utility. This is because maximizing the  sum throughput utility reduces to opportunistically maximizing the sum throughput in each time slot.
Particularly, considering $D_\text{max} > R_\text{max}$, the  data admission policy  \eqref{eq:sumThroughputPolicy} implies that  during steady state the data buffers in all BNs will fluctuate near $V+D_\text{max}$. Hence, all BNs will have almost equal weights in data link scheduling problem \eqref{prob:LinkSchedule} and accordingly, our problem reduces to maximizing the sum  rate as in  \cite{Mishra_Sum}. 
\begin{figure}
\centering
\includegraphics[width =0.6 \textwidth]{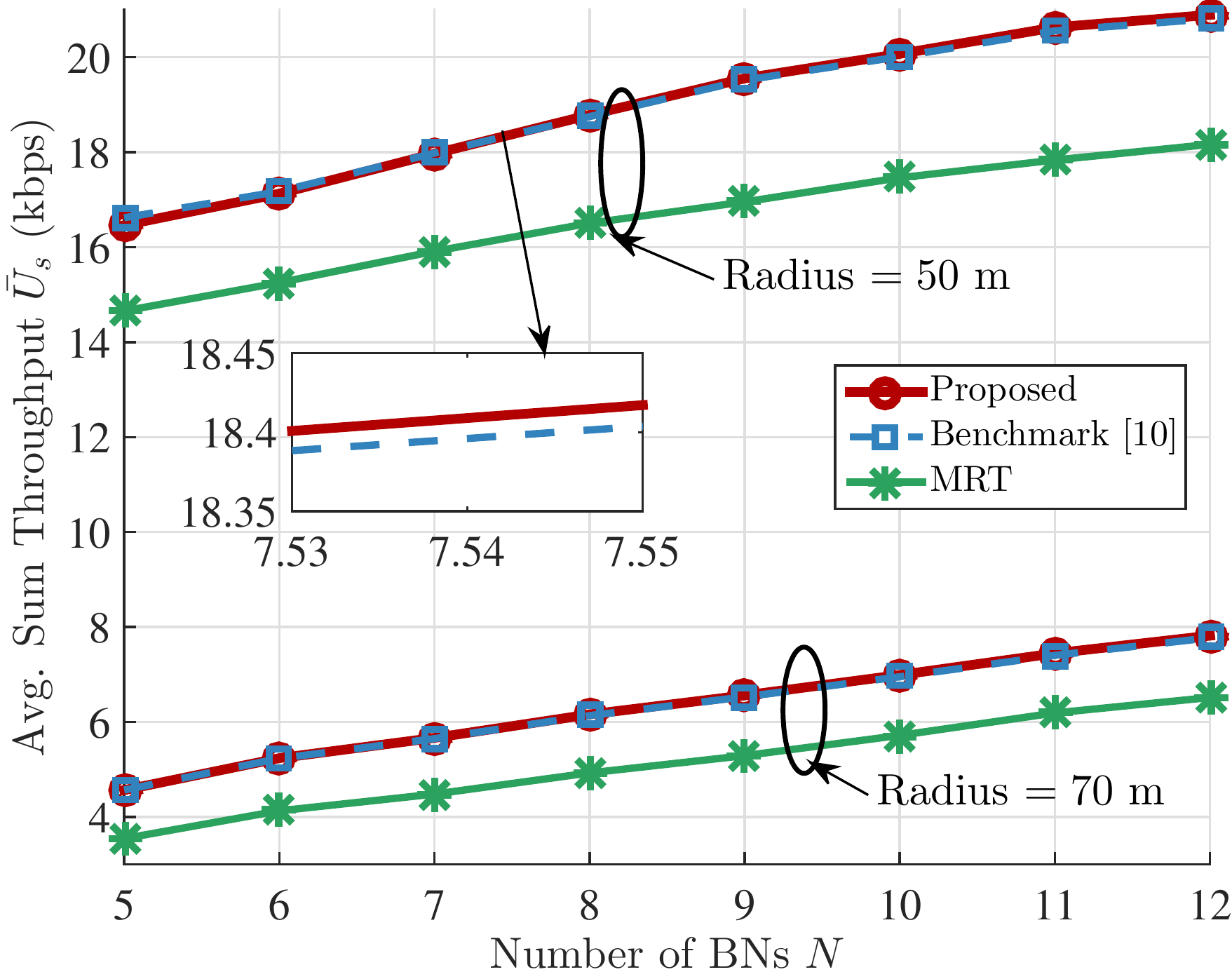}
\caption{ Sum throughput utility versus the number of the BNs $N$ for the proposed algorithm, the benchmark in \cite{Mishra_Sum} and MRT. BNs are distributed randomly in a circular area with Radius  $\{50, 70\}$ m and $V = 10^7$.} 
\label{fig:sumVsN}
\end{figure}

Considering the  utilities, $U_s(t)$, $U_p(t)$ and $U_c(t)$, Fig. \ref{fig:utilityComp} compares the average throughput and the received energy of the BNs in a BCN with 4 BNs. In this BCN, $\text{BN}_i, i \in \{1,2,3,4\}$, are located at distances 18 m, 22 m, 30 m and 34 m of the Reader, respectively.  According to Fig. \ref{fig:BNthroughput}, the sum throughput utility behaves opportunistically where, for instance, $\text{BN}_1$ achieves almost $250\%$ higher throughput compared to $\text{BN}_4$. Whereas, with the common throughput utility all BNs achieve equal throughputs at a cost of about $20\%$  sum throughput reduction compared to the other two utilities. The received energy of the BNs in Fig. \ref{fig:BNpower} is in harmony with Fig. \ref{fig:BNthroughput}. That is, with the sum and proportional throughput utilities, $\text{BN}_1$ receives the highest portion of the energy transmitted by the Reader. However, with the  common throughput utility the farther   BNs  receive more energy.  Particularly, considering the common throughput utility, $\text{BN}_4$ receives $230\%$ more energy compared to $\text{BN}_1$. While  the path loss of the link between $\text{BN}_4$ and the Reader is $(34/18)^3 \approx 6.7$ times higher than path loss of the link between $\text{BN}_1$ and the Reader.

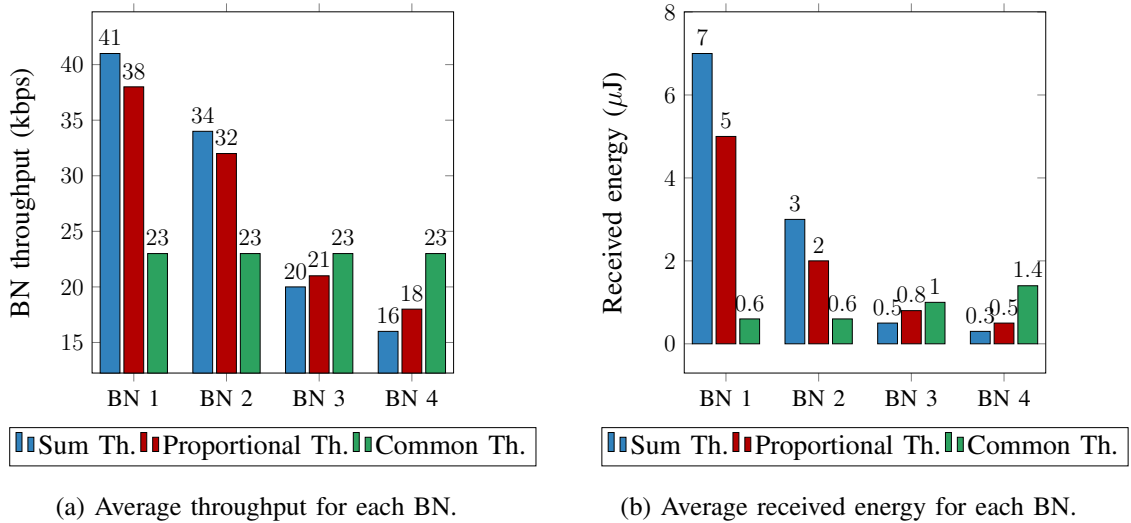
\begin{figure}
\centering
\begin{subfigure}[b]{0.45\textwidth}
\definecolor{r1}{rgb}{   0.1922,    0.5098,    0.7412 } 
\definecolor{r2}{rgb}{0.702, 0, 0  } 
\definecolor{r3}{rgb}{0.1725,    0.6353,    0.3725} 
\begin{tikzpicture}
\begin{axis}[
width=8cm,height=8cm,
    ybar,
    enlargelimits=0.15,
    legend style={at={(0.5,-0.15)},
      anchor=north,legend columns=-1},
    ylabel={\large BN throughput (kbps) },
    symbolic x coords={BN 1,BN 2,BN 3, BN 4},
    xtick=data,
    nodes near coords,
    nodes near coords align={vertical},
    ]

%

\addplot[fill=r1] coordinates {(BN 1, 41) (BN 2, 34 ) (BN 3,20) (BN 4,16) };
\addplot[fill=r2]coordinates {(BN 1,38) (BN 2,32) (BN 3,21) (BN 4,18) };
\addplot[ fill=r3] coordinates {(BN 1,23) (BN 2,23) (BN 3,23) (BN 4,23) };

\legend{\large Sum Th.,\large Proportional Th.,\large Common Th. }
\end{axis}
\end{tikzpicture}
\caption{Average throughput for each BN.}\label{fig:BNthroughput}
\end{subfigure}
~
\begin{subfigure}[b]{0.45\textwidth}
\definecolor{r1}{rgb}{   0.1922,    0.5098,    0.7412 } 
\definecolor{r2}{rgb}{0.702, 0, 0  } 
\definecolor{r3}{rgb}{0.1725,    0.6353,    0.3725} 
\begin{tikzpicture}
\begin{axis}[
width=8cm,height=8cm,
    ybar,
    enlargelimits=0.15,
    legend style={at={(0.5,-0.15)},
      anchor=north,legend columns=-1},
    ylabel={\large Received energy ($\mu$J) },
    symbolic x coords={BN 1,BN 2,BN 3, BN 4},
    xtick=data,
    nodes near coords,
    nodes near coords align={vertical},
    ]

\addplot[fill=r1 ] coordinates {(BN 1, 7) (BN 2, 3 ) (BN 3,0.5) (BN 4,0.3) };
\addplot[fill=r2]coordinates {(BN 1, 5) (BN 2, 2 ) (BN 3,0.8) (BN 4,0.5) };
\addplot[fill=r3] coordinates {(BN 1, 0.6) (BN 2,0.6) (BN 3,1) (BN 4,1.4) };

\legend{\large Sum Th.,\large  Proportional Th.,\large Common Th. }
\end{axis}
\end{tikzpicture}
\caption{Average received energy for each BN.}\label{fig:BNpower}
\end{subfigure}
\caption{Utility functions comparison. The distances of  the $\text{BN}_n,\; n= \{1,2,3,4\}$,  to the Reader are set to  $18$ m, $22$ m, $30$ m and $34$ m, respectively, and $V = 10^7$.  }
\label{fig:utilityComp}
\end{figure}

\subsection{Finite Blocklength Analysis of the Proposed Scheme} \label{sec:finite}
To simplify the analysis, we presented the results for the cases with sufficiently long codewords where the Shannon's capacity, i.e., \eqref{eq:Rn}, gives an appropriate approximation of the achievable rates. However, depending on the Reader energy budget and the number of nodes, the BCN may be of interest in the cases with short packets. For this reason, in Fig. \ref{fig:finiteLen}, we investigate the effect of the finite length codewords on the communication range under the proposed policy. {Particularly, we define the communication range as the maximum distance of the BNs to the Reader such that the BNs achieve a specified common throughput with certain error probability. Here, we consider  BCNs with BNs located circularly around the Reader.} Using the fundamental results of \cite{Polyanskiy, Makki2014, Makki2016, Haghifam} on the achievable rates of short packets, we replace \eqref{eq:Rn} with \cite[Theorem 45]{Polyanskiy},
\begin{align}\label{eq:RnMod}
R_n(t) \approx W\Bigg[\log\left(1+\mathcal{S}_n(t) \right) -\sqrt{  {1\over {L}}\bigg[1- { \bigg(1+ \mathcal{S}_n(t)\bigg)^{-2}  }\bigg]   } \mathcal{Q}^{-1}(\psi)\Bigg],
\end{align}
which approximates the achievable rate with the finite length codewords. In  \eqref{eq:RnMod}, $L$ denotes the packet length, $\mathcal{Q}^{-1}(.)$ is the inverse $Q$-function and  $\psi$ is the maximum codeword error probability of the decoder. Hence, the second term inside the brackets in \eqref{eq:RnMod} is notable only in the case of codewords with finite length. Then, letting $L\rightarrow \infty$ \eqref{eq:RnMod} is simplified to \eqref{eq:Rn} for the cases with asymptotically long codewords. {Note that  replacing  \eqref{eq:Rn} with the modified rate function \eqref{eq:RnMod}, the results in Theorem \ref{th:optimal} are still applicable and our data admission policy remains optimal. However, Algorithm \ref{alg:ED} is an approximate solution for data link scheduling Problem \eqref{prob:LinkSchedule} with the modified rate function \eqref{eq:RnMod}. }

 Considering  the common throughputs $\bar{U}_c(t) = \{ 3 ,5 \}$ kbps, Fig. \ref{fig:finiteLen} shows the communication range versus the numbers of the BNs $N$. This figure is plotted for codeword lengths $L = \{\infty, 10^4, 10^3, 10^2\}$ and the error probability  $\psi = 10^{-3}$. As seen in Fig. \ref{fig:finiteLen}, in harmony with intuition, while the Shannon's capacity-based evaluations are tight for the cases with long codewords, with finite length codewords the communication range decreases.  {This communication range reduction is almost constant for different number of BNs $N$.} Moreover, the  communication range reduction is more significant  as the common throughput decreases. This is because, with lower signal to interference plus noise  ratio,  the second term inside the brackets in \eqref{eq:RnMod} is more dominant. Finally,  Fig. \ref{fig:finiteLen} shows the tradeoff between communication range and number of BNs. According to this figure, the communication range decreases almost linearly with the number of BNs. 
\begin{figure}
\centering
\includegraphics[width =0.6 \textwidth]{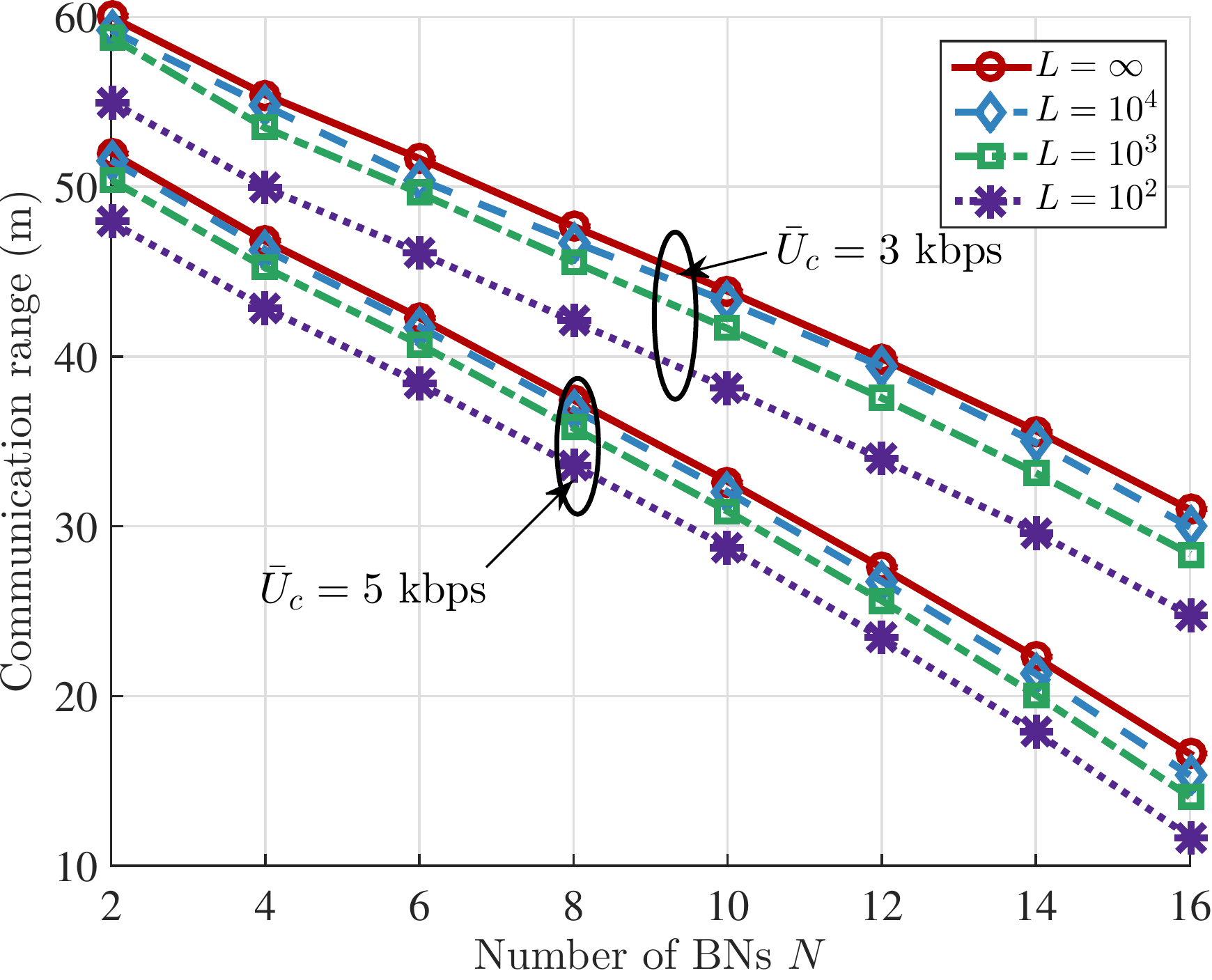}
\caption{ Communication range versus the number of the BNs, $N$ for  codeword lengths $L = \{\infty, 10^4, 10^3, 10^2\}$, target common throughput set to $\bar{U}_c = \{3,5\}$ kbps, and $V = 10^7$. }
\label{fig:finiteLen}
\end{figure}

Considering different codeword lengths $L = \{10^2, 10^3\}$ and  Rician $K$-factors $K=\{0, 1,10,100\}$, Fig. \ref{fig:commVsPerror} shows the average common throughput utility $\bar{U}_c$ versus the maximum  error probability $\psi$ in \eqref{eq:RnMod}. According to this figure, considering different $K$-factors, the utility increases with the error probability almost logarithmically. However, the common throughput utility is more sensitive to the  error probability with low codeword lengths. As the codeword length increases, the sensitivity to the error probability decreases.
Therefore, at lower targeted error probabilities  short codewords can significantly degrade the utility  which should be considered in BCN design. 
Moreover, the common throughput increases as  the Rician $K$-factor increases. That is because the deterministic component of the channel $\bm{h}_n^d(t)$ becomes more dominant as $K$ increases. This relative increment saturates at higher $K$-factors.   Finally, the common throughput is more sensitive to the Rician $K$-factor with longer codewords. 

\begin{figure}
\centering
\includegraphics[width =0.6 \textwidth]{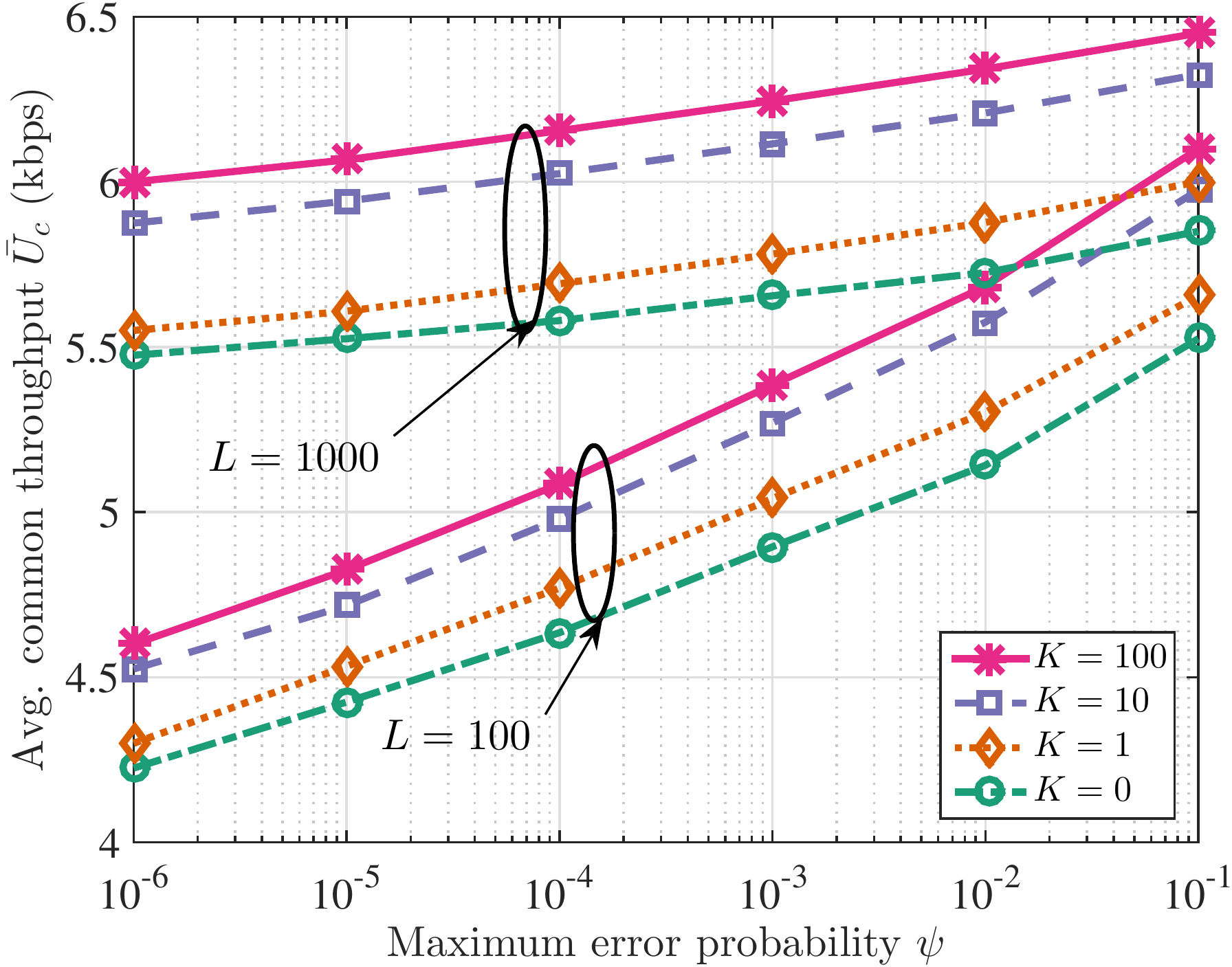}
\caption{ Common throughput utility versus the maximum codeword error probability, $\psi$ for  codeword lengths $L = \{10^2, 10^3\}$, Rician $K$-factors $K\in\{0, 1,10,100\}$ and $V = 10^7$. } 
\label{fig:commVsPerror}
\end{figure}

\section{Conclusion}\label{sec:conclude}
This paper studied data and energy scheduling in a monostatic BCN with a multi-antenna Reader and multiple BNs. The BNs adopt buffers to store their admitted  data  before transmission to the Reader. We proposed data link control and data admission policies for  maximizing the average value of  different  utility functions, including the sum, the proportional and the common throughput utilities.  Through simulation comparisons, we showed the superiority of our proposed scheme compared to state-of-the-art works.   Specifically, we showed that the sum throughput utility under our proposed policy achieves the  maximum sum channel rate in  \cite{Mishra_Sum}, while stabilizing the data buffers. Whereas, our proposed policy for optimizing the common throughput improves the result in \cite{Mishra_Min}. Moreover, using the results on the achievable rates of finite blocklength codewords, we  studied the system performance in the cases with short packets. 

As demonstrated, the proposed policies achieve optimal utility and stabilize the data buffers in the BNs with  few iterations of the data link control algorithm. Moreover, considering the common throughput utility, adopting buffers in BNs enables  more efficient data scheduling and, hence, improves the common throughput.  
Finally, according to our  analysis, the finite length of the codewords  affects the communication range significantly, specifically at low signal to interference plus noise ratios. Also, the utilities are more sensitive to maximum error probability at short codeword lengths, which should be carefully compensated for in the BCN design. 
\appendices
\section{Proof of Lemma \ref{lem:upperbound}}\label{proof:upperbound}
Consider the following chain of inequalities,
\begin{align}\label{eq:uppBoundChain}
\begin{split}
\operatorname{E}\Big\{   L(t+1) - L(t)\Big | \bm{Q}(t) \Big\} &= \frac{1}{2}\sum_{n \in \mathcal{N}} \operatorname{E}\Big\{Q_n^2(t+1)-Q_n^2(t)\Big| \bm{Q}(t)\Big \} \\
&\overset{(a)}{=}\frac{1}{2}\sum_{n \in \mathcal{N}} \operatorname{E}\bigg\{  \big([Q_n(t) - R_n(r)]^+\big)^2 +D^2_n(t) + \\
&\qquad\qquad 2D_n(t)[Q_n(t) - R_n(r)]^+ -Q_n(t)^2\Big| \bm{Q}(t)\bigg\}\\
&\overset{(b)}{\leq}\frac{1}{2}\sum_{n \in \mathcal{N}}\operatorname{E}\bigg\{R_n(t)^2-2Q_n(t)R_n(t)+D_n(t)^2+2D_n(t)Q_n(t) \Big| \bm{Q}(t)\bigg\}\\
&\overset{(c)}{\leq} B + \sum_{n \in \mathcal{N}}\operatorname{E}\bigg\{Q_n(t) (D_n(t) - R_n(t))\Big| \bm{Q}(t)\bigg\},
\end{split}
\end{align}
where $(a)$ results from   \eqref{eq:QueueEvolve}. Inequality $(b)$ in \eqref{eq:uppBoundChain} holds because $\big([Q_n(t) - R_n(t)]^+\big)^2 \leq \big(Q_n(t) - R_n(t)\big)^2 $ and $[Q_n(t) - R_n(t)]^+ \leq Q_n(t)$. Moreover, inequality $(c)$ comes from 
 $$\frac{1}{2}\sum_{n \in \mathcal{N}} \big(R_n^2(t) + D_n^2(t)\big) \leq  \frac{N}{2}\big(R_{max}^2 + D_\text{max}^2(t)\big) = B.$$ Finally, \eqref{eq:upperbound} in Lemma \ref{lem:upperbound} is proved if we subtract $V\operatorname{E}\big\{\mathcal{U}(t) |\bm{Q}(t)\big\}$  from both sides of \eqref{eq:uppBoundChain}.

\section{Proof of Theorem \ref{th:optimal}}\label{proof:theoremOpt}
\paragraph{Proof of the first claim} To prove the first claim of the theorem, we show that under the solution of Problem \eqref{prob:upperBound}, we  always have $D_n(t) = 0$ if $Q_n(t) \geq V$. To avoid clutter, we omit the time slot index $t$. Assume that the buffer level in the $l$-th BN exceeds $V$, i.e., $Q_l \geq V$.
We define two data admission vectors  $\bm{D}^1$  and $\bm{D}^0$ that only differ in the $l$-th element. Specifically,   we have $D_n^1 = D_n^0, \;\forall n \neq l$, $D_l^0 = 0$ and $D_l^1 > 0$.
 With  slightly modified notation, we define $\Delta_u^i \triangleq \Delta_u\big(\bm{f}, \bm{G}, \bm{\alpha}, \bm{D}^i\big), \;i = \{0,1\}$. Note that we have  modified the notation to emphasize that $\Delta_u(.)$ in \eqref{eq:upperbound} is a function $\bm{f}, \bm{G}, \bm{\alpha}$ and $\bm{D}$. Accordingly, we have 
\begin{align}\label{eq:bufferUpp}
\begin{split}
\Delta_u^1  - \Delta_u^0 &\overset{(a)}{=} Q_l(t)D^1_l - V\Big(U_\text{g}\big(\bm{D}^1\big)- U_\text{g}\big(\bm{D}^0\big)\Big) \\
&\overset{(b)}{ \geq} Q_l(t)D^1_l - VD^1_l   = \big(Q_l(t)- V\big)D^1_l  \geq 0,
\end{split}
\end{align}
where, $\text{g}\in\{\text{s},\text{p},\text{c} \}$ indicates the sum, proportional and common throughput utility functions. The equality $(a)$ follows from the definition of $\Delta_u(.)$ in \eqref{eq:upperbound}.  Moreover, $(b)$ results from the fact that the utility functions $U_\text{g}(\bm{D}),\; \text{g}\in\{\text{s},\text{p},\text{c} \}$, are Lipschitz continuous, such that we have
\begin{align}
U_\text{g}\big(\bm{D}^1\big)- U_\text{g}\big(\bm{D}^0\big) \leq \sum_{n \in \mathcal{N}} \big(D^1_n - D_n^0\big) = D^1_l.
\end{align}
Then, \eqref{eq:bufferUpp} implies that $\Delta_u^1 \geq \Delta_u^0$, and hence, a data admission vector with $D_l > 0$ is not optimal. Accordingly, we conclude that $Q_n(t)$ will never exceed $V+ D_\text{max}$.

 \paragraph{Proof of the second claim} The second claim can be proved following the Lyapunov optimization method in \cite[Chapter 4]{Neely}. According to \cite[Theorem 4.5]{Neely}, there is an optimal stationary policy, denoted by random-only-policy, which is only a function of $\bm{h}_n(t),\;\forall n$. Under the random-only-policy, we have $\operatorname{E}\Big\{\mathcal{U}(t)|\bm{Q}(t)\Big\} = \mathcal{U}^\star$ and $\operatorname{E}\Big\{Q_n(t) (D_n(t) - R_n(t))\Big\} \leq 0$. Plugging the random-only-policy in $\Delta_u(t)$ and using \eqref{eq:upperbound}, we have
\begin{align}\label{eq:optWonly}
\Delta_p(L(t)) \leq B - V\mathcal{U}^\star,
\end{align}
where $\Delta_p(L(t)) $ is evaluated under the solution of Problem \eqref{prob:upperBound}. Note that \eqref{eq:optWonly} holds since the solution of Problem \eqref{prob:upperBound} minimizes $\Delta_u(t)$, and hence, $\Delta_p(L(t)) $ under the solution of Problem \eqref{prob:upperBound} is not greater than $\Delta_u(t)$ under all alternative solutions, including the random-only-policy. Averaging both sides of \eqref{eq:optWonly} over $t = 0, \ldots, T-1$, we obtain
\begin{align}\label{eq:optWonly2}
{1 \over T}\sum_{n \in \mathcal{N}}Q^2_n(t)-{1 \over T}\sum_{n \in \mathcal{N}}Q^2_n(0)  -V{1 \over T}\sum_{t = 0}^{T-1}\operatorname{E}\{\mathcal{U}(t)\} \leq B - V\mathcal{U}^\star,
\end{align}
and, rearranging the terms in \eqref{eq:optWonly2} and taking $\lim_{T\rightarrow \infty}$, we have
 \begin{align}\label{eq:optWonly3}
\lim_{T\rightarrow \infty}{1 \over T}\sum_{t = 0}^{T-1}\operatorname{E}\{\mathcal{U}(t)\} \geq \mathcal{U}^\star - {B\over V}.
\end{align}
%
\section{Proof of Proposition \ref{prop:conv}}\label{proof:probConv}
We first show that the value of the function $\sum Q_n(t) R_n(t)$ is non-decreasing in each iteration. To avoid clutter, we omit the time slot index $t$. Moreover, to emphasize the dependence on $\bm{f},\bm{G}$ and $\bm{\alpha}$,
we use the notations $\tilde{R}(\bm{f},\bm{G},\bm{\alpha}, \bm{\gamma})$, $\vardbtilde{R}(\bm{f},\bm{G},\bm{\alpha}, \bm{\gamma}, \bm{y})$ and $R_q(\bm{f},\bm{G}, \bm{\alpha})$  instead of $\tilde{R}(t, \bm{\gamma})$, $\vardbtilde{R}(t, \bm{\gamma}, \bm{y})$ and $\sum Q_n(t) R_n(t)$, respectively. Let $\bm{f}^i, \bm{G}^i$ and $\bm{\alpha}^i$ denote the value of $\bm{f}, \bm{G}$ and $\bm{\alpha}$ at the beginning of the $i$-th iteration. Then, we have 
\begin{align}\label{eq:ChainConv}
\begin{split}
R_q(\bm{f}^i,\bm{G}^i, \bm{\alpha}^i) &\overset{(a)}{=} \tilde{R}(\bm{f}^i,\bm{G}^i,\bm{\alpha}^i, \bm{\gamma}^\star)  \overset{(b)}{=} \vardbtilde{R}(\bm{f}^i,\bm{G}^i,\bm{\alpha}^i, \bm{\gamma}^\star, \bm{y}^\star)\\
& \overset{(c)}{\leq}\vardbtilde{R}(\bm{f}^{i+1},\bm{G}^i,\bm{\alpha}^i, \bm{\gamma}^\star, \bm{y}^\star) \overset{(d)}{\leq}\vardbtilde{R}(\bm{f}^{i+1},\bm{G}^i,\bm{\alpha}^{i+1}, \bm{\gamma}^\star, \bm{y}^\star)\\
& \leq \max_{\bm{y}} \vardbtilde{R}(\bm{f}^{i+1},\bm{G}^i,\bm{\alpha}^{i+1}, \bm{\gamma}^\star, \bm{y})   \overset{(e)}{=}\tilde{R}(\bm{f}^{i+1},\bm{G}^i,\bm{\alpha}^{i+1}, \bm{\gamma}^\star)\\
& \leq \max_{\bm{\gamma}} \tilde{R}(\bm{f}^{i+1},\bm{G}^i,\bm{\alpha}^{i+1}, \bm{\gamma}) \overset{(f)}{=} R_q(\bm{f}^{i+1},\bm{G}^{i}, \bm{\alpha}^{i+1}) \\
& \overset{(g)}{\leq} R_q(\bm{f}^{i+1},\bm{G}^{i+1}, \bm{\alpha}^{i+1}), 
\end{split}
\end{align}
 where $(a)$ holds because $\bm{\gamma}^\star$ maximizes $\tilde{R}(\bm{f}^i,\bm{G}^i,\bm{\alpha}^i, \bm{\gamma})$ and the maximum value reduces to $R_q(\bm{f}^i,\bm{G}^i, \bm{\alpha}^i)$. Similarly, $(b)$ follows from the fact that $\bm{y}^\star$ maximizes $\vardbtilde{R}(\bm{f}^i,\bm{G}^i,\bm{\alpha}^i, \bm{\gamma}^\star, \bm{y})$, and the maximum value reduces to $\tilde{R}(\bm{f}^i,\bm{G}^i,\bm{\alpha}^i, \bm{\gamma}^\star)$. Inequalities $(c)$ and $(d)$ hold since $\bm{f}^{i+1}$ and $\bm{\alpha}^{i+1}$ maximize $\vardbtilde{R}(.)$ while the other variables are fixed. The equalities $(e)$ and $(f)$ are concluded similar to $(b)$ and $(a)$, respectively. Finally, $(g)$ follows since $\bm{G}^{i+1}$ maximizes $R_q(\bm{f}^{i+1},\bm{G}, \bm{\alpha}^{i+1})$. According to \eqref{eq:ChainConv}, the value of the sum $\sum Q_n(t) R_n(t)$  is non-decreasing over successive iterations. Moreover, since $\sum Q_n(t) R_n(t)$ is upper-bounded, the iterations converge. 
\bibliographystyle{IEEEtran}
\bibliography{Bibs}

\end{document}